\documentclass[11pt]{article}
\usepackage[T1]{fontenc}
\usepackage[utf8]{inputenc}
\usepackage[english]{babel}
\usepackage{amssymb}
\usepackage{amsmath}
\usepackage{amsthm}
\usepackage{mathtools}
\usepackage{booktabs}
\usepackage{enumitem}
\usepackage{microtype}
\usepackage{needspace}
\usepackage[margin=1in]{geometry}
\usepackage[hidelinks]{hyperref}
\usepackage{color}

\newtheorem{theorem}{Theorem}
\newtheorem{definition}{Definition}

\newtheorem{lemma}[theorem]{Lemma}

\DeclareMathOperator*{\argmax}{arg\,max}

\newcommand{\cX}{\mathcal{X}}
\newcommand{\cY}{\mathcal{Y}}
\newcommand{\cU}{\mathcal{U}}
\newcommand{\cC}{\mathcal{C}}
\newcommand{\cP}{\mathcal{P}}

\newcommand{\cF}{\mathcal{F}}
\newcommand{\cE}{\mathcal{E}}
\newcommand{\T}{\mathsf{T}}
\newcommand{\E}{\mathbb{E}}
\newcommand{\Pp}{\mathbb{P}}
\newcommand{\1}{\mathbf{1}}
\newcommand{\wt}{\widetilde}
\newcommand{\wh}{\widehat}
\newcommand{\bu}{\boldsymbol{u}}
\newcommand{\bU}{\boldsymbol{U}}
\newcommand{\bX}{\boldsymbol{X}}

\newcommand{\bx}{\boldsymbol{x}}
\newcommand{\by}{\boldsymbol{y}}
\newcommand{\bz}{\boldsymbol{z}}
\newcommand{\bxbar}{\overline{\boldsymbol{x}}}
\newcommand{\Pe}{P_{\mathrm e}}

\newcommand{\Ilm}{I_{\mathrm{LM}}}
\newcommand{\Clm}{C_{\mathrm{LM}}}
\newcommand{\Cm}{C_{\mathrm{M}}}
\newcommand{\pos}{\operatorname{pos}}

\usepackage{amsmath}
\usepackage{graphicx}

\newcommand{\wc}[1]{%
	\overset{\smash{\rotatebox[origin=c]{180}{$\widehat{\phantom{#1}}$}}}{#1}%
}

\title{Diagnosing the Refuted Mismatched Decoding \\ Converse for Binary-Input Channels}
\author{Jonathan Scarlett \medskip \\ National University of Singapore}
\date{}

\begin{document}

\maketitle

\begin{abstract}
	We revisit the claimed converse theorem for mismatched decoding over binary-input discrete memoryless channels (Balakirsky, 1995) and the subsequent work refuting this converse via a numerical counter-example (Scarlett, Somekh-Baruch, Martinez, and Guill\'en i F\`abregas, 2015).  A notable gap in the existing understanding is that no concrete flaw was identified in Balakirsky's analysis. In this paper, Balakirsky's proof is ``diagnosed'', and it is demonstrated that (i) his \emph{permutation lemma} is flawed regardless of the correctness of his \emph{combinatorial approximation lemma}; (ii) there are at least two specific incorrect steps regarding restricted codes and selector sequences, and these appear to be unlikely to admit a local repair; and (iii) the proposition in which these incorrect steps are used is false in general.
	
	This work is AI-assisted, and a full declaration is given at the end of the paper.
\end{abstract}

\section{Introduction}

The mismatched decoding problem concerns the goal of reliable communication when the decoding rule is fixed and possibly suboptimal \cite{ScarlettEtAl2020}.  This is a fundamental problem in information theory that is not only of interest in its own right, but also recovers notions such as zero-error capacity and zero-undetected error capacity as special cases. 
In the early literature on mismatched decoding (see Section \ref{sec:related} for a more detailed summary), the vast majority of results concerned achievable rates for discrete memoryless channels (DMCs) via random coding, notably including the Generalized Mutual Information (GMI) via the i.i.d.~ensemble \cite{KaplanShamai1993}, and the so-called LM rate\footnote{Seemingly an acronym for ``lower [bound on the] mismatch [capacity]''.} via the constant-composition ensemble \cite{Hui1983,CsiszarNarayan1995}.  On the other hand, single-letter \emph{upper (converse) bounds} on the mismatch capacity largely arose only recently \cite{KangarshahiGuillen2021,SomekhBaruch2022}, and were severely lacking in the early literature, with one clear exception: Balakirsky \cite{Balakirsky1995} reported that the LM rate is tight for binary-input DMCs. 

Roughly two decades later, a numerical counter-example was given that refuted this binary-input converse \cite{ScarlettEtAl2015}; specifically, it was shown that an achievable rate based on superposition coding can be strictly higher than the LM rate.\footnote{The relevant code from \cite{ScarlettEtAl2015} has been moved to a different location \cite{ScarlettCode2026}.  Moreover, one of the lemmas in \cite{ScarlettEtAl2015} contains a minor error that turns out to be inconsequential and easily fixed; see Appendix \ref{sec:our_correction} of the present paper for the details.} However, this work did not provide any clarification on where Balakirsky's argument fails, with the authors writing \emph{``...despite considerable effort, we have been unable to understand the analysis given in [Balakirsky's paper] in sufficient detail to identify any major errors therein''}.

In this paper, we perform a ``diagnosis'' of Balakirsky's paper \cite{Balakirsky1995} and identify two specific proof steps that appear to be incorrect and unlikely to be repairable.  We show that:
\begin{itemize}
    \item[(i)] the steps themselves \cite[Eqs.~(75) and (76)]{Balakirsky1995} regarding selector sequences and restricted codes are not true in general (see Section \ref{sec:false-steps});
    \item[(ii)] the proposition containing those steps \cite[Proposition~6.5]{Balakirsky1995}, which asymptotically compares the error probability with and without a suitably-defined code restriction, is not true in general, suggesting that the flawed steps may not be ``locally repairable'' (see Section \ref{sec:prop65-superposition}).
    \item[(iii)] the lemma containing that proposition \cite[Lemma~4.5]{Balakirsky1995} (termed the \emph{permutation lemma}) is not true in general, suggesting that the flawed proposition may not be repairable either (see Section \ref{sec:perm-false}).
\end{itemize}
A natural fourth item in this list would be that the main theorem in \cite{Balakirsky1995} is not true in general, which is precisely what was already established in \cite{ScarlettEtAl2015}.  We note that our presentation follows a different logical ordering from the above list: We first establish (iii), followed by (i) and then (ii).  All three claims were developed with significant AI assistance; a more detailed declaration is given at the end of the paper.

\subsection{Problem Setup}

{\bf The mismatched decoding problem.} Let $W=\{W_x(y)\}$ be a discrete memoryless channel from a finite input alphabet $\cX$ (assumed to be $\{0,1\}$ throughout this paper) to a finite output alphabet $\cY$, and let $q=\{q(x,y)\}$ be a non-negative decoding metric.  A length-$n$ code is a collection
\begin{equation}
    \cC_n=\{\bx^{(1)},\ldots,\bx^{(M)}\}\subseteq \cX^n,
    \label{eq:code}
\end{equation}
and the decoder estimates the message according to the \emph{maximum metric rule}:
\begin{gather}
    \widehat{m} = \argmax_{j = 1,\dotsc,M} q^n(\bx^{(j)},\by), \label{eq:dec-rule} \\ 
    q^n(\bx,\by)=\prod_{i=1}^n q(x_i,y_i).
    \label{eq:metric-product}
\end{gather}
As in \cite{Balakirsky1995}, we assume that ties are counted as errors.  

In the following, we will use the generic symbol $K$ for a channel that might later be identified with $W$, $V$ (introduced below), or both.   For a code $\cC_n$ and an $n$-letter channel $K_n$, we write
\begin{equation}
    \Pe^{(n)}(K_n;\cC_n)
    =\max_{\bx\in\cC_n}
        \sum_{\by}K_n(\by|\bx)
        \1\left\{ q^n(\bx,\by) \leq \max_{\bx'\in\cC_n\setminus\{\bx\}} q^n(\bx',\by) \right\}
    \label{eq:maximal-error}
\end{equation}
for the maximal mismatched-decoding error probability.  When the code is clear from context, we simply write $\Pe^{(n)}(K_n)$.  For a memoryless channel $\{K_x(y)\}$, we write $K^n(\by|\bx)=\prod_{i=1}^nK_{x_i}(y_i)$. %, which is distinct from the combinatorial channel $K_n$ corresponding to \eqref{eq:combinatorial-channel}.

Throughout the paper, given an input distribution $P$ and channel $W$, the resulting output marginal is denoted by $PW(y) = \sum_x P(x)W_x(y)$, and similarly for $PV(y)$.  Moreover, the function $\log(\cdot)$ has base $e$, and information measures are in units of nats.

{\bf Achievable rates.} For the pair $(W,q)$, a rate $R$ is said to be \emph{achievable} if, for any $\delta > 0$, there exists a sequence of codes $\mathcal{C}_n$ with at least $e^{n(R-\delta)}$ codewords attaining $\lim_{n \to \infty} \Pe^{(n)}(W^n;\cC_n) = 0$.  The \emph{mismatch capacity} $\Cm$ is defined to be the maximum of all achievable rates.  We are primarily interested in the achievable \emph{LM rate}, which is defined for a fixed input distribution $P$ by \cite{Hui1983,CsiszarNarayan1995}
\begin{equation}
    \Ilm(P) =\min_{\substack{ V:\,PV=PW,\, \\ 
        \E_{P\times V}[\log q(X,Y)] \geq \E_{P\times W}[\log q(X,Y)]}} I(P,V),
    \label{eq:LM-primal}
\end{equation}
and the optimized LM rate is $\Clm=\max_P\Ilm(P)$, giving $\Clm \le \Cm$.  Its achievability is established by a standard analysis of constant-composition random codes via the method of types.

We will also use the dual form of the standard superposition coding rate \cite{ScarlettEtAl2016,SomekhBaruch2015}, which depends on an auxiliary alphabet $\mathcal{U}$ and will be described in more detail in Section \ref{sec:summary_ce}.  For now, we only introduce the relevant rate expression.  
% Let $K(y|x)$ be a generic channel paired with the metric $q(x,y)$, and fix a joint distribution $Q_{UX}$. 
For $s\geq0$, $\rho \in [0,1]$, and a real-valued function $a(u,x)$, and a joint distribution $Q_{UX}$, define
\begin{align}
    R^{\rm SC}_1(s,a)
    & =\E\left[
        \log\frac{q(X,Y)^s e^{a(U,X)}}
        {\sum_{\bar x}Q_{X|U}(\bar x|U)
            q(\bar x,Y)^s e^{a(U,\bar x)}} \right],
    \label{eq:SC-dual-B1}\\
    R^{\rm SC}_0(\rho,s,a;R_1)
    & =\E\left[
        \log\frac{(q(X,Y)^s e^{a(U,X)})^\rho}
        {\sum_{\bar u}Q_U(\bar u) \left(\sum_{\bar x}Q_{X|U}(\bar x|\bar u)
            q(\bar x,Y)^s e^{a(\bar u,\bar x)}\right)^\rho} \right] -\rho R_1,
    \label{eq:SC-dual-B0}
\end{align}
where $(U,X,Y)\sim Q_{UX}\times W$.  Then, the rate $R = R_0 + R_1$ is achievable whenever $R_1 \le \sup_{s \ge 0,a(\cdot)} R^{\rm SC}_1(s,a)$ and $R_0 \le \sup_{s \ge 0,\rho \in [0,1],a(\cdot)} R^{\rm SC}_0(\rho,s,a;R_1)$ \cite{ScarlettEtAl2016,SomekhBaruch2015}.  (Note that $s$ and $a(\cdot)$ may be optimized separately in the two bounds, but $Q_{UX}$ must be common to both.) Following the counter-example in \cite{ScarlettEtAl2015}, we will apply these formulas to the pair channel $(W^2,q^2)$, and similarly to $(V^2,q^2)$, by treating each pair as a single channel symbol (see Section \ref{sec:summary_ce} for details).

{\bf Other useful notations and definitions.} 
For an input distribution $P\in\cP(\cX)$, a \emph{$P$-constant-composition} code is one in which every codeword has empirical distribution $P$.  Given a conditional type $V\in\cP(\cY|\cX)$ and a sequence $\bx$ of type $P$, define the corresponding conditional type class
\begin{equation}
    \T_V^n(\bx)
    =\big\{\by\in\cY^n:\widehat P_{\by|\bx}=V\big\}
    \label{eq:conditional-type-class}
\end{equation}
where $\widehat P_{\by|\bx}$ denotes the empirical conditional distribution.  Balakirsky's \emph{combinatorial channel} associated with $V$ is defined as \cite[Eq.~(28)]{Balakirsky1995}
\begin{equation}
    V_n(\by|\bx)
    =\frac{\1\{\by\in\T_V^n(\bx)\}}{|\T_V^n(\bx)|},
    \label{eq:combinatorial-channel}
\end{equation}
and $W_n(\by|\bx)$ is defined analogously.  
Here and subsequently, following \cite[Convention~4.2]{Balakirsky1995}, whenever exact conditional types are used, we restrict $n$ to an infinite common subsequence on which the relevant distributions are indeed valid types.  Since the distributions used in the present paper are rational, such a subsequence is obtained simply by taking $n$ to be multiples of a common denominator.
% To simplify notation, we will interchangeably treat $P,V,W$ as general distributions and as types corresponding to length $n$. with the understanding that more generally the closest type approximations can be used without fundamentally altering any of the arguments. 

For a sequence $\{a_n\}$ of non-negative numbers, define the following asymptotic notation \cite[Definition 2.1]{Balakirsky1995}:
\begin{equation}
	\Upsilon(\{a_n\}) =
	\begin{cases}
		1 & \liminf_{n\to\infty}a_n>0 \\
		0 & \textrm{otherwise}.
	\end{cases}
	\label{eq:upsilon}
\end{equation}
In particular, an inequality of the form $\Upsilon(\{a_n\})\leq \Upsilon(\{b_n\})$ states that $a_n$ being eventually bounded away from zero implies the same for $b_n$.

While we framed the problem in terms of maximum-metric decoding, one may equivalently use an \emph{additive distortion function} $d$ and select the codeword \emph{minimizing} the total distortion
$d^n(\bx,\by) = \sum_{i=1}^n d(x_i,y_i)$.  In particular, for arbitrary functions $a(x)$ and $b(y)$, the following choice gives an equivalent decoding rule under constant-composition codes \cite{CsiszarNarayan1995}:
\begin{equation}
	d_x(y) = a(x) + b(y) - \log q(x,y).
	\label{eq:d_log_q}
\end{equation}
Indeed, the contribution $\sum_i a(x_i)$ is the same for every codeword of a given composition, while $\sum_i b(y_i)$ only depends on the output sequence.  We will often use distortion-based notation for consistency with \cite{Balakirsky1995}.\footnote{When $\mathcal{X} = \{0,1\}$ and $\min_{x,y} q(x,y) > 0$, we can readily choose $a(\cdot)$ and $b(\cdot)$ to match the conventions $d_x(y) \in [0,d_{\max}]$ and $d_0(y)=0$ ($\forall y$) from \cite{Balakirsky1995}: Setting $b(y)=\log q(0,y)-a(0)$ ensures $d_0(y)=0$, and choosing $a(1)-a(0)$ sufficiently large makes all $d_1(y)$ non-negative.}
Given a distortion measure $d_x(y)$, under a $P$-constant-composition code and channel $K$, we will write
\begin{equation}
    d(P,K)=\sum_{x,y}P(x)K_x(y)d_x(y). 
    \label{eq:average-distortion}
\end{equation}
for the corresponding average distortion.

\subsection{Related Work} \label{sec:related}

By far the most relevant works are those deriving the LM rate \cite{Hui1983,CsiszarNarayan1995}, Balakirsky's claimed converse \cite{Balakirsky1995}, the numerical counter-example \cite{ScarlettEtAl2015}, and those deriving achievable rates based on superposition coding \cite{ScarlettEtAl2016,SomekhBaruch2015}.  We provide a broader (but still brief) outline of the related literature as follows; see \cite{ScarlettEtAl2020} for a detailed survey.

{\bf Basic achievable rates.} The LM rate was obtained via constant-composition random coding by both Hui \cite{Hui1983} and Csisz\'ar and Narayan \cite{CsiszarNarayan1995}.  Another widely-studied achievable rate is the generalized mutual information (GMI), which arises from i.i.d.~random coding, e.g., see \cite{KaplanShamai1993}.  These rates were further studied by Merhav \emph{et al.}~\cite{MerhavEtAl1994} and Ganti \emph{et al.}~\cite{GantiEtAl2000}, including extensions to continuous-alphabet channels via cost-constrained random coding.  Refined asymptotics and similar cost-constrained ensembles were subsequently studied by Scarlett \emph{et al.}~\cite{ScarlettEtAl2014}.

{\bf Improved rates.} It is well known that the LM rate can be strictly smaller than the mismatch capacity in general.  One approach to improving it is the product-channel construction of Csisz\'ar and Narayan \cite{CsiszarNarayan1995}, in which the LM rate is applied to higher-order product channels.  They conjectured that the limiting multi-letter rate obtained in this manner equals the mismatch capacity.  This conjecture remains open for the standard maximum-metric mismatched decoder, but it was shown to be true for erasures-only metrics in \cite{CsiszarNarayan1995}.  More recently, Molina and Guill\'en i F\`abregas \cite{MolinaGuillen2026} established the corresponding conjecture for mismatched stochastic decoding.

A distinct approach to improving the LM rate is to use multi-user coding techniques.  Lapidoth \cite{Lapidoth1996} initiated this idea, obtaining an improved achievable rate by embedding the single-user problem into a mismatched multiple-access channel.  Superposition coding was subsequently studied (concurrently) by Scarlett \emph{et al.}~\cite{ScarlettEtAl2016} and Somekh-Baruch \cite{SomekhBaruch2015}, leading to further achievable rates that can strictly exceed the LM rate.  

{\bf Claimed converse and refutation.} As outlined in the introduction, Balakirsky \cite{Balakirsky1995} claimed that the LM rate equals the mismatch capacity for binary-input discrete memoryless channels.  Scarlett \emph{et al.}~\cite{ScarlettEtAl2015} refuted this claim by combining the two kinds of improvement outlined above; superposition coding was applied over a second-order product channel to obtain an achievable rate strictly exceeding the LM rate.  The present paper revisits Balakirsky's converse proof and identifies concrete failures in the underlying argument.

{\bf Other converse results.} Single-letter converse bounds on the mismatch capacity were obtained relatively recently by Asadi Kangarshahi and Guill\'en i F\`abregas \cite{KangarshahiGuillen2021} and Somekh-Baruch \cite{SomekhBaruch2022}.  These works give non-trivial upper bounds that can lie strictly below the matched capacity, complementing the improved achievable rates described above.  In particular, for the counter-example of \cite{ScarlettEtAl2015}, there is a strict separation between the LM rate, the improved achievable rate via superposition coding, the best known single-letter upper bound, and the matched capacity.

Beyond the above-outlined results on the Csisz\'ar--Narayan conjecture, other multi-letter results have also been established.  In particular, Somekh-Baruch derived a general multi-letter information-spectrum formula for the mismatch capacity \cite{SomekhBaruch2015General}, and subsequently developed further multi-letter converse results for discrete memoryless channels \cite{SomekhBaruch2018} including a ``soft converse'' mirroring the Csisz\'ar--Narayan conjecture for rational-valued additive decoding metrics.

\section{Claimed Converse and Counter-Example} \label{sec:claimed}

\subsection{Summary of Claimed Converse}

After a standard reduction that permits focusing on constant-composition codes (with composition denoted by $P$), Balakirsky's argument boils down to two main lemmas concerning the true channel $W$ and another channel $V$ satisfying the constraints of \eqref{eq:LM-primal}:
\begin{itemize}
	\item \textbf{Combinatorial approximation lemma \cite[Lemma 4.4]{Balakirsky1995}:} This lemma replaces the memoryless channels $V$ and $W$ by combinatorial channels $V_n$ and $W_n$ (defined in \eqref{eq:combinatorial-channel})  that are uniform over the corresponding conditional type classes. The lemma claims that, for the purpose of whether the maximal decoding error stays bounded away from zero, this replacement does not change the asymptotic behavior.
	\item \textbf{Permutation lemma \cite[Lemma 4.5]{Balakirsky1995}:} This lemma couples an output $y\in T_V^n(x)$ to an output $y'\in T_W^n(x)$ by suitable permutations of coordinates, using the condition $PV=PW$ from \eqref{eq:LM-primal}. The lemma claims that if decoding under $V_n$ has non-vanishing maximal error, then decoding under $W_n$ must also have non-vanishing maximal error.
\end{itemize}
A more detailed summary of Balakirsky's argument will be given in \eqref{eq:Bal-steps} below.

The permutation lemma will be more prominent in this paper.  It is invoked for $V$ satisfying 
\begin{align}
    PV & =PW, \label{eq:bal-marginal}\\
    I(P,V) & \leq I(P,W), \label{eq:bal-mutual-information}\\
    d(P,V) & \leq d(P,W), \label{eq:bal-distortion}
\end{align}
and for every sequence of $P$-constant-composition codes, it claims that
\begin{equation}
    \Upsilon\big(\{\Pe^{(n)}(V_n)\}\big)
	\stackrel{\textrm{claimed}}{\leq} 
    \Upsilon\big(\{\Pe^{(n)}(W_n)\}\big),
    \label{eq:permutation-lemma-claim}
\end{equation}
where $V_n$ and $W_n$ are the exact-type combinatorial channels in \eqref{eq:combinatorial-channel}, $\Pe^{(n)}(\cdot)$ is the maximal error probability in \eqref{eq:maximal-error}, and the asymptotic $\Upsilon$ notation is defined in \eqref{eq:upsilon}. 

By combining the combinatorial approximation lemma and permutation lemma with the weak converse for $V$, and optimizing over all such $V$, the following composition-dependent upper bound on the rate $R$ is stated in \cite{Balakirsky1995}:
\begin{equation}
    R \stackrel{\textrm{claimed}}{\leq} \min_{V:\,PV=PW,\,d(P,V)\leq d(P,W)}I(P,V) = \Ilm(P).
    \label{eq:bal-fixed-P-bound}
\end{equation}
Further optimizing over $P$ gives the claimed equality $\Cm \stackrel{\textrm{claimed}}{=} \Clm$ for binary-input channels \cite{Balakirsky1995}.

\subsection{Summary of the Counter-Example} \label{sec:summary_ce}

The counter-example of \cite{ScarlettEtAl2015} refuting this converse uses $\mathcal{X} = \{0,1\}$, $\mathcal{Y} = \{0,1,2\}$, and
\begin{equation}
    W=\begin{bmatrix}
        0.97 & 0.03 & 0\\
        0.10 & 0.10 & 0.80
    \end{bmatrix},
    \qquad
    q=\begin{bmatrix}
        1 & 1 & 1\\
        1 & 0.5 & 1.36
    \end{bmatrix}.
    \label{eq:counterexample-Wq}
\end{equation}
By combining theoretical simplifications with numerical calculations, it was shown that
\begin{equation}
    0.136874\leq \Clm\leq 0.136900,
    \qquad
    C_{\mathrm{SC}}^{(2)}\geq 0.137998,
    \label{eq:published-gap}
\end{equation}
where $C_{\mathrm{SC}}^{(2)}$ is an improved achievable rate obtained via superposition coding \cite{ScarlettEtAl2016,SomekhBaruch2015}.  Specifically, the construction applies superposition coding to the product channel $(W^2,q^2)$, noting that any achievable rate for this pair amounts to half that rate being achievable for $(W,q)$.  This product channel has input alphabet $\cX^2=\{00,01,10,11\}$, and the auxiliary alphabet for superposition coding is chosen as $\cU=\{0,1\}$.  

The constant-composition superposition ensemble is specified by a joint
distribution $Q_{UX^2}$, length $m$, and rates $(R_0,R_1)$, and consists of the following:
\begin{itemize}
    \item $e^{mR_0}$ auxiliary sequences $\bu$, referred to as
    \emph{cloud centers}, each drawn uniformly from the type class
    corresponding to $Q_U$;
    \item for each cloud center $\bu$, $e^{mR_1}$ conditionally generated
	codewords $\bx$, referred to as \emph{satellites}, each drawn uniformly from the conditional type class corresponding to $Q_{X^2|U}$ given $\bu$.
\end{itemize}
Note that $m$ represents the number of uses of $(W^2,q^2)$, and we set $m=n/2$ 
to match the block length of $n$ with respect to $(W,q)$. 
The counter-example sets
\begin{equation}
    P=(p_0,p_1)=(0.749,0.251),
    \label{eq:P-counterexample}
\end{equation}
and chooses $Q_{UX^2}$ via its marginals and conditionals as follows:
\begin{align}
    Q_U & =(1-p_1^2,p_1^2), \label{eq:QU}\\
    Q_{X^2|U=0}
    & =\frac{1}{1-p_1^2}
    (p_0^2,p_0p_1,p_0p_1,0), \label{eq:QXU0}\\
    Q_{X^2|U=1}
    & =(0,0,0,1). \label{eq:QXU1}
\end{align}
This choice has $X^2$-marginal exactly $P^2$.  Hence, every length-$m$
satellite codeword has composition $P^2$, which implies that after unfolding the symbol
pairs, every resulting length-$n$ binary codeword has composition $P$.

\subsection{Further Useful Observations} \label{sec:further-useful}

In \cite[Eq.~(51)]{ScarlettEtAl2015} it is established that the rate pair $(R_0,R_1) = (0.0356005,0.2403966)$ lies in the achievable superposition coding region (see also \cite{ScarlettCode2026} for the code).     
For the diagnosis below, it is convenient to ``back off'' slightly and set
\begin{equation}
    R_0=0.035 \qquad R_1=0.240,
    \label{eq:backed-off-rates}
\end{equation}
corresponding to the binary-symbol rate
\begin{equation}
    R=\frac{R_0+R_1}{2}=0.1375.
    \label{eq:binary-rate}
\end{equation}
Because $(R_0,R_1)$ is strictly inside the superposition coding achievable region for $(W^2,q^2)$, the ensemble has exponentially vanishing average error probability under $W$  \cite{ScarlettEtAl2016,SomekhBaruch2015}.  Expurgating a fixed fraction (e.g., half) of the messages gives \emph{maximal error probability} tending to zero while preserving the composition and asymptotic rate.

We will use the following explicit ``test channel'' satisfying the constraints in \eqref{eq:bal-marginal}--\eqref{eq:bal-distortion}:
\begin{equation}
	V=\begin{bmatrix}
		\frac{6449550}{7490000} & \frac{475449}{7490000} & \frac{565001}{7490000}\\
		0.425 & 0.0001 & 0.5749
	\end{bmatrix}.
	\label{eq:test-channel-V}
\end{equation}
While the top row may appear unnatural, it comes from solving the constraint $PV=PW$ for $V_0$, namely $V_0(y)=\frac{PW(y)-P(1)V_1(y)}{P(0)}$ (see also \cite{ScarlettCode2026} for numerical verification that $PV = PW$ here).  In approximate numerical form, it gives $V_0 \approx (0.861088, 0.063478, 0.075434)$.

A direct evaluation of mutual information then gives the following (see \cite{ScarlettCode2026} for numerical verification), where $I(P,W)$ is included only for the purpose of verifying \eqref{eq:bal-mutual-information}:
\begin{align}
    I(P,V) & =0.1368941713\dotsc
    \label{eq:IPV-value}\\
    I(P,W) & =0.4205303151\dotsc
    \label{eq:IPW-value}
\end{align}
Moreover, the difference in decoding metrics simplifies due to most values of $\log q(x,y)$ being zero under the metric in \eqref{eq:counterexample-Wq}:
\begin{align}
    &\E_{P\times V}[\log q(X,Y)]
        -\E_{P\times W}[\log q(X,Y)]
    \notag\\
    &\hspace{2cm}
    	= p_1 \bigg( \big(V_1(1)-W_1(1)\big)\log\frac{1}{2} + \big(V_1(2) - W_1(2)\big) \log(1.36)  \bigg) \notag \\
    &\hspace{2cm}
        =0.251 \big(0.0999\log 2-0.2251\log 1.36\big)
    \notag\\
    &\hspace{2cm}
        =7.67995\dotsc\times 10^{-6}>0.
    \label{eq:metric-feasible-V}
\end{align}
Hence, $V$ satisfies the marginal and metric constraints defining the LM rate in \eqref{eq:LM-primal} for $(P,W,q)$.  Under the distortion representation \eqref{eq:d_log_q}, we can equivalently write \eqref{eq:metric-feasible-V} as $d(P,V)<d(P,W)$, in accordance with \eqref{eq:bal-distortion}.

\section{Diagnosing the Claimed Converse}
\label{sec:diagnosis}

Before proceeding, it is useful to summarize the overall logic of Balakirsky’s argument:\footnote{\cite[Section VI-F]{Balakirsky1995} instead states an asymptotic \emph{equality} (using $\Upsilon$ notation) at the $\le$ step, but the non-asymptotic upper bound is sufficient here and follows trivially from the fact that the restriction only removes competing codewords.}
\begin{equation}
	P_e^{(n)}(V^n)
	\overset{\mathrm{L4.4}}{\approx}
		\underbrace{
		P_e^{(n)}(V_n)
		\overset{\mathrm{P6.5}}{\approx}
		\widehat P_e^{(n)}(V_n)
		\overset{\mathrm{P6.7}}{\lesssim}
		\widehat P_e^{(n)}(W_n)
		\le
		P_e^{(n)}(W_n)
}	_{\text{Lemma 4.5 (permutation lemma)}}
\overset{\mathrm{L4.4}}{\approx}
P_e^{(n)}(W^n), \label{eq:Bal-steps}
\end{equation}
where all $\approx$ and $\lesssim$ indicate statements that are claimed to hold upon applying $\Upsilon(\cdot)$ (see \eqref{eq:upsilon}) to both sides, $(\cdot)^n$ and $(\cdot)_n$ respectively denote memoryless channels and combinatorial channels (\emph{cf.}, \eqref{eq:combinatorial-channel}), and $\widehat P_e^{(n)}$ is a notion of restricted error probability that we will define in Section \ref{sec:false-steps}.  Note also that `L` indicates a lemma from \cite{Balakirsky1995}, and `P' indicates a proposition.

\subsection{Failure of the Permutation Lemma } \label{sec:perm-false}

We first state a lemma that serves to separate the failure of the permutation lemma from any possible issue in the combinatorial approximation lemma (summarized at the start of Section \ref{sec:claimed}).

\begin{lemma}
For the choices of $P$, $V$, and $W$ in \eqref{eq:P-counterexample}, \eqref{eq:counterexample-Wq}, and \eqref{eq:test-channel-V} respectively, there is a sequence of $P$-constant-composition codes of rate approaching $0.1375$ such that
\begin{gather}
    \lim_{n \to \infty} \Pe^{(n)}(W_n) = 0,
    \label{eq:exact-W-good}\\
    \liminf_{n\to\infty}\Pe^{(n)}(V_n) >0.
    \label{eq:exact-V-bad}
\end{gather}
Consequently, the permutation lemma (stated in \eqref{eq:permutation-lemma-claim}) is false.
\end{lemma}
\begin{proof}
We use the superposition code described in Section \ref{sec:summary_ce}, including the standard expurgation argument to convert from average error to maximal error (without impacting the asymptotic rate), thus having exponentially vanishing maximal error under the memoryless channel $W^n$.  

We note that the probability that $W^n$ produces the exact conditional type $W$ is only polynomially small.  More precisely, a standard conditional-type probability bound \cite[Ch.~2]{CsiszarKorner1981} gives the following for every codeword $\bx$ of type $P$ (where the constant $6$ comes from $|\cX| \cdot |\cY|$):
\begin{equation}
    W^n\big(\T_W^n(\bx) \,|\, \bx \big)\geq (n+1)^{-6}.
    \label{eq:type-prob-lower}
\end{equation}
Since the conditional $n$-letter channel law given this event is precisely $W_n$, it follows that
\begin{equation}
    \Pe^{(n)}(W_n) \leq (n+1)^{6}\Pe^{(n)}(W^n) \to 0.
    \label{eq:memoryless-to-exact-W}
\end{equation}
For the combinatorial channel $V_n$ in \eqref{eq:combinatorial-channel}, let the message $M$ be uniform, and let $Y^n$ be the channel output. Since every conditional type class $\T_V^n(\bx)$ (among all $\bx$ with composition $P$) has the same size, and every possible output has type $PV$, we have
\begin{align}
    H(Y^n|M) & =\log |\T_V^n(\bx)|, \label{eq:HY-given-M}\\
    H(Y^n) & \leq \log |\T_{PV}^n|. \label{eq:HY-upper}
\end{align}
An application of standard bounds on type class sizes \cite[Ch.~2]{CsiszarKorner1981} then yields
\begin{equation}
    I(M;Y^n) = H(Y^n) - H(Y^n | M) \leq nI(P,V)+O(\log n).
    \label{eq:exact-V-information}
\end{equation}
Fano's inequality therefore gives the following for an arbitrary decoder:
\begin{equation}
    \liminf_{n\to\infty}\Pe^{(n)}(V_n) \geq 1-\frac{I(P,V)}{R},
    \label{eq:Fano-exact-V}
\end{equation}
which is strictly positive in view of \eqref{eq:binary-rate} and \eqref{eq:IPV-value}.  This applies in particular to the mismatched decoder.  Equations \eqref{eq:memoryless-to-exact-W} and \eqref{eq:Fano-exact-V} contradict \eqref{eq:permutation-lemma-claim}.
\end{proof}

The significance of this lemma is that the combinatorial approximation lemma is not needed to see the failure, allowing us to focus our attention on the permutation lemma and the intermediate results that it depends on.

\subsection{Two Specific Incorrect Steps in \cite[Proposition 6.5]{Balakirsky1995}}
\label{sec:false-steps}

We now diagnose two specific steps based on selector sequences and restrictions, namely \cite[Eqs.~(75) and (76)]{Balakirsky1995}, showing that they are incorrect and unlikely to be repairable.  To describe these steps, we need to introduce the notion $\widehat P_e^{(n)}$ of restricted error probability that we used in \eqref{eq:Bal-steps}.

{\bf Defining the restricted code.} Since the definition of the restricted code is rather technical, we first attempt to give some intuition.  The high-level idea is to extract the competing codewords whose empirical behavior on the selected coordinates is ``typical'' with respect to the \emph{channel difference} $V - W$ (with suitable clipping and normalization described below).  One might then hope that this restriction does not substantially alter whether decoding under $V_n$ has non-vanishing error probability, since sufficiently many ``typical'' error-causing competing codewords should remain (the $\overset{\mathrm{P6.5}}{\approx}$ step in \eqref{eq:Bal-steps}).  As part of the proof of this claim, Balakirsky modifies the mechanism for selecting coordinates so that it depends only on the output sequence and not on the transmitted codeword, and this is where we identify the flawed steps (see \eqref{eq:Bal-75} and \eqref{eq:Bal-76} below).

We now proceed with the technical details.  For two binary-input channels $V$ and $W$ satisfying $PV=PW$, Balakirsky defines the following (denoted by $Y_x^+$ in \cite[Eq.~(54)]{Balakirsky1995}:
\begin{equation}
    \cY_x^+=\{y:V_x(y)>W_x(y)\}.
    \label{eq:Y-plus}
\end{equation}
Note that since we are in the binary-input setting, the condition $PV = PW$ can be written as $p_0( V_0(y) - W_0(y) ) = -p_1 ( V_1(y) - W_1(y) )$, which implies the following whenever $\min\{p_0,p_1\} > 0$ (which is assumed throughout Balakirsky's analysis \cite[Convention 4.2]{Balakirsky1995}):
\begin{equation}
	\cY_0^+ = \cY_1^- \quad \text{ and } \quad \cY_1^+ = \cY_0^-. \label{eq:Ysets}
\end{equation}
 For each $y\in\cY_0^+\cup\cY_1^+$, we can therefore define \cite[Eq.~(66)]{Balakirsky1995}
 \begin{equation}
 	x(y) = (\text{unique input symbol such that $y\in\cY_{x(y)}^+$}). \label{eq:x(y)}
 \end{equation}
If $y \notin \cY_0^+\cup\cY_1^+$ (i.e., $V_x(y) = W_x(y)$ for both values of $x \in \{0,1\}$) then the value of $x(y) \in \{0,1\}$ is assigned arbitrarily, and its choice is inconsequential for the analysis.\footnote{In particular, when $V_x(y) = W_x(y)$ the two terms in the $\min\{\cdot\}$ in \eqref{eq:S-zero} are identical, and \eqref{eq:S-one} is identically zero.}
 
Next, the distribution $\Delta V_x$ appearing in \cite[Eqs.~(55) and (69)]{Balakirsky1995} can be written more compactly as follows (using the definitions of $U_x(y)$ and $k_x$ therein):
\begin{equation}
    \Delta V_x(y) =\frac{[V_x(y)-W_x(y)]^+} {\sum_{y'}[V_x(y')-W_x(y')]^+},
    \label{eq:Delta-V}
\end{equation}
where $[\cdot]^+ = \max\{0,\cdot\}$.  Thus, $\Delta V_x(y)$ represents a suitably clipped and normalized ``channel difference'' between $V$ and $W$.

Consider a positive tolerance sequence $\{\alpha_n\}$ satisfying \cite[Eq.~(71)]{Balakirsky1995} 
\begin{equation}
    \alpha_n\to0,
    \qquad
    \alpha_n\sqrt n\to\infty.
    \label{eq:alpha-65}
\end{equation}
Given a (transmitted or competitor) codeword $\bxbar$, an output sequence $\by$, and a \emph{selector sequence $\bz\in\{0,1\}^n$} (treated as arbitrary for now and specified later), define \cite[Eqs.~(67)--(68)]{Balakirsky1995} 
\begin{align}
    k_{ab}(y) & =\sum_{j:z_j=1} \1\{\bar x_j=b,\,y_j=y,\,x(y)=a\},
    \label{eq:kaby}\\
    k_{ab} & =\sum_{y\in\cY_a^+}k_{ab}(y).
    \label{eq:kab}
\end{align}
These are the empirical counts that we alluded to earlier, and the following notion of a \emph{restricted code} \cite[Definition 6.4]{Balakirsky1995} formalizes the notion of only maintaining codewords whose empirical behavior is ``typical'' with respect to $\Delta V$.

\begin{definition} \label{def:restriction}
	Given an output sequence $\by$ and a selector sequence $\bz$, the \emph{$\Delta V$-restricted code} $\cC_n(\by,\bz)$ consists of the codewords $\bxbar$ satisfying
	\begin{equation}
		\left|k_{ab}(y)-k_{ab}\Delta V_a(y)\right| \leq \alpha_n k_{ab}
		\label{eq:Bal-restriction}
	\end{equation}
	for all $a,b \in \{0,1\}$ and all $y \in \mathcal{Y}$.
\end{definition}

{\bf Defining the selector distributions.} 
The (conditional) distribution on $\bz$ initially introduced in \cite[Proposition 6.5]{Balakirsky1995} depends on both $\bx$ and $\by$.  Specifically, using our compact $[\cdot]^+$ notation, \cite[Eq.~(65)]{Balakirsky1995} can be written as
\begin{align}
    \wh Q_{x,y}(1) & =\frac{[V_x(y)-W_x(y)]^+}{V_x(y)},
    \label{eq:hatQ-one}\\
    \wh Q_{x,y}(0) & =1-\wh Q_{x,y}(1),
    \label{eq:hatQ-zero}
\end{align}
with the ratio set to zero when $V_x(y)=0$.  The corresponding type-based $n$-letter law is then defined as follows given $(\bx,\by)$:
\begin{equation}
	\wh{Q}_n(\bz | \bx,\by)=\frac{\mathbf{1}\{\wh P_{\bz\mid\bx,\by}=\wh Q\}}{|T^n_{\wh Q}(\bx,\by)|},
	\label{eq:Qhat-n}
\end{equation}
where $\wh{P}_{\bz\mid\bx,\by}$ is the empirical conditional type, and $T^n_{\wh Q}(\bx,\by)$ denotes the conditional type class corresponding to conditional type $\wh Q$.

Balakirsky replaces the selector distribution $\wh{Q}_n$ by a different selector distribution $\wt{Q}_n$ that depends on $\by$ but not on the transmitted word.  Namely, his construction sets \cite[p.~1898 bottom]{Balakirsky1995}
\begin{equation}
    \wt Q_y(z)=\wh Q_{x(y),y}(z)
    \label{eq:tildeQ-single}
\end{equation}
with $x(y)$ defined in \eqref{eq:x(y)}, and lets $\wt{Q}_n(\bz|\by)$ be uniform on the corresponding conditional type class given $\by$, similarly to \eqref{eq:Qhat-n}.

{\bf Statement and refutation of \cite[Eq.~(75)]{Balakirsky1995}.} To formally state the claimed step replacing $\wh Q$ by $\wt Q$, let $\cC_n(\by,\bz)$ denote the restricted subcode from Definition \ref{def:restriction}, and define
\begin{equation}
    D_n(\bx,\by|\bz) =\min_{\bxbar\in\cC_n(\by,\bz)\setminus\{\bx\}} d^n(\bxbar,\by).
    \label{eq:restricted-D}
\end{equation}
Then, \cite[Eq.~(75)]{Balakirsky1995} claims that
\begin{align}
    \wh P_{\mathrm e}^{(n)}(\bx,V_n)
    &:= \sum_{\by,\bz}
    V_n(\by|\bx)\wh{Q}_n(\bz|\bx,\by)
    \mathbf{1}\{\bx\in\cC_n(\by,\bz)\}
    \mathbf{1}\{d^n(\bx,\by)\ge D_n(\bx,\by|\bz)\} \label{eq:Phat} \\
    & \stackrel{\textrm{claimed}}{\geq}
    \sum_{\by,\bz}
        V_n(\by|\bx)\wt{Q}_n(\bz|\by)
        \1\{\bx\in\cC_n(\by,\bz)\}
        \1\{d^n(\bx,\by)\geq D_n(\bx,\by|\bz)\}
     \label{eq:Bal-75} 
\end{align}
The stated justification for \eqref{eq:Bal-75} is that the $\by$-only selector places stronger restrictions on the incorrect codewords.  We speculate that this belief stemmed from a related observation, namely, that the $\by$-only selector has higher selection probabilities (i.e., $\widetilde Q_y(1)\geq \widehat Q_{x,y}(1)$ for all $(x,y)$), so one might expect it to constrain the competing codewords more strongly in \eqref{eq:Bal-restriction}.  However, the restriction in \eqref{eq:Bal-restriction} is not monotone in the number of selected coordinates; rather, it depends on the relative empirical frequencies.  The following lemma refutes the claimed step via a specific counter-example.

\begin{lemma} \label{lem:wrong-step-75}
	There exists a channel-metric pair $(W,q)$, input distribution $P$,
	auxiliary channel $V$, block length $n$, and  $P$-constant-composition code 
	for which the inequality in \eqref{eq:Bal-75} is false (while the conditions \eqref{eq:bal-marginal}--\eqref{eq:bal-distortion} are all satisfied).
\end{lemma}
\begin{proof}
	The proof is somewhat lengthy and uses a different (non-asymptotic) counter-example than the one from Section \ref{sec:summary_ce}, so it is deferred to Appendix \ref{app:pf-75}.
\end{proof}

	{\bf Statement and refutation of \cite[Eq.~(76)]{Balakirsky1995}.} 
	We give a second concrete failure in the proof of
	\cite[Proposition~6.5]{Balakirsky1995}, which appears to be  an even more significant obstruction than the first.  
	Unlike the failure of \cite[Eq.~(75)]{Balakirsky1995}, it is based on the same choices of $P$, $(W,q)$ and $V$ from the earlier parts (i.e., from Sections~\ref{sec:summary_ce}--\ref{sec:further-useful}), though superposition coding does not play a role here.
	
	As noted above, Balakirsky initially uses the
	$(\bx,\by)$-dependent selector $\wh Q_n$ in
	\eqref{eq:hatQ-one}--\eqref{eq:Qhat-n}, and subsequently replaces it by the
	$\by$-only selector $\wt Q_n$ defined via \eqref{eq:tildeQ-single}, i.e., $\wt Q_y(z)=\wh Q_{x(y),y}(z)$.  
	After this replacement, \cite[Eq.~(76)]{Balakirsky1995} claims that one can
	choose $\{\alpha_n\}$ satisfying \eqref{eq:alpha-65} such that
	\begin{equation}
		\lim_{n \to \infty} \bigg( \sum_{\by,\bz}
		V_n(\by|\bx)\wt Q_n(\bz|\by)
		\1\{\bx\in\cC_n(\by,\bz)\} \bigg)
		\stackrel{\textrm{claimed}}{=} 1,
		\label{eq:Bal-76}
	\end{equation}
	where $\cC_n(\by,\bz)$ is the $\Delta V$-restricted code in
	Definition~\ref{def:restriction}.  That is, the transmitted codeword is claimed to remain in the restricted code with probability approaching one under the $\by$-only selector. Under the original $(\bx,\by)$-dependent selector (see \eqref{eq:Qhat-n}), the analogous statement holds deterministically by construction, since the exact conditional type of the selector sequence forces the relevant empirical frequencies to match the prescribed $\Delta V$ proportions.  The justification given for \eqref{eq:Bal-76} is that it holds by a similar argument to \cite[Lemma~4.1]{Balakirsky1995}, which in turn is based on a standard argument via the method of types and concentration.
	
	We show that, for the parameters of the counter-example used throughout this paper, the limiting statement in \eqref{eq:Bal-76} not only fails, but the left-hand side eventually equals zero.
	
	\begin{lemma}
		\label{lem:false-76}
		Consider the choices of $P$, $(W,q)$, and $V$ in
		\eqref{eq:P-counterexample}, \eqref{eq:counterexample-Wq}, and
		\eqref{eq:test-channel-V}, respectively.  For every sequence $\bx$ of type $P$,
		and every sequence $\{\alpha_n\}$ satisfying $\alpha_n\to0$, it holds for all sufficiently large $n$ that\footnote{As with our other results, this implicitly restricts to the common subsequence of $n$ in which all relevant (conditional) types are well-defined.}
		\begin{equation}
			\sum_{\by,\bz}
			V_n(\by|\bx)\wt Q_n(\bz|\by)
			\1\{\bx\in\cC_n(\by,\bz)\}
			=0.
			\label{eq:Bal-76-zero}
		\end{equation}
		Consequently, the claim in \eqref{eq:Bal-76} is false. 
	\end{lemma}
	
	\begin{proof}
		For the channels $W$ and $V$ in \eqref{eq:counterexample-Wq} and \eqref{eq:test-channel-V} (with $\mathcal{Y} = \{0,1,2\}$), a direct calculation via the definitions in \eqref{eq:Y-plus} and \eqref{eq:Delta-V} gives
		\begin{equation}
			\cY_0^+=\{1,2\},
			\qquad
			\Delta V_0(1)=\frac{999}{3250}.
			\label{eq:76-Yplus-Delta}
		\end{equation}
		Moreover, the definition of $x(\cdot)$ in \eqref{eq:x(y)} gives $x(2)=0$.  We also observe that $W_0(2)=0$ while $V_0(2)>0$, and hence
		the definition of the $\by$-only selector distribution in
		\eqref{eq:hatQ-one} and \eqref{eq:tildeQ-single} gives
		\begin{equation}
			\wt Q_2(1)
			=\wh Q_{0,2}(1)
			=\frac{V_0(2)-W_0(2)}{V_0(2)}
			=1.
			\label{eq:76-Qtilde-2}
		\end{equation}
		Hence, for every $\bz$ in the support of $\wt Q_n(\cdot|\by)$, we have the following implication:
		\begin{equation}
			y_j=2 \quad\Longrightarrow\quad z_j=1.
			\label{eq:76-y2-selected}
		\end{equation}
		
		Fix any $\by\in\T_V^n(\bx)$ and any $\bz$ satisfying
		$\wt Q_n(\bz|\by)>0$.  We proceed to test whether the transmitted codeword itself,
		i.e., $\bxbar=\bx$, satisfies the restriction
		\eqref{eq:Bal-restriction}.  This condition is required to hold for all $(a,b,y)$, but we will identify a violation from just a single choice:
		\begin{equation}
			a=0,\qquad b=1,\qquad y=1.
		\end{equation}
		Although we focus on $y=1$, we still need to understand both $k_{01}(1)$ and $k_{01}(2)$ since they both contribute to the term $k_{01}$ on the right-hand side of \eqref{eq:Bal-restriction}.  Starting with $k_{01}(2)$, since $\by$ has exact conditional type $V$ given $\bx$, the number of
		coordinates having $(x_j,y_j)=(1,2)$ is exactly
		$n p_1V_1(2)$.  All of these coordinates are selected due to
		\eqref{eq:76-y2-selected}, and hence
		\begin{equation}
			k_{01}(2)
			=np_1V_1(2)
			=np_1 \cdot 0.5749.
			\label{eq:76-k012}
		\end{equation}
		On the other hand, no matter how many coordinates having
		$(x_j,y_j)=(1,1)$ are selected, we have the analogous upper bound
		\begin{equation}
			k_{01}(1)
			\leq np_1V_1(1)
			=np_1 \cdot 0.0001.
			\label{eq:76-k011}
		\end{equation}
		Since $\cY_0^+=\{1,2\}$, the definition of $k_{ab}$ in \eqref{eq:kab} gives
		$k_{01}=k_{01}(1)+k_{01}(2)$, and therefore
		\begin{align}
			\frac{k_{01}(1)}{k_{01}}
			&\leq
			\frac{V_1(1)}{V_1(1)+V_1(2)}
			\notag\\
			&=
			\frac{0.0001}{0.0001+0.5749}
			=\frac{1}{5750}.
			\label{eq:76-ratio-bound}
		\end{align}
		Combining this with \eqref{eq:76-Yplus-Delta} yields the deterministic lower
		bound
		\begin{align}
			\left|
			\frac{k_{01}(1)}{k_{01}}
			-\Delta V_0(1)
			\right|
			&\geq
			\frac{999}{3250}-\frac{1}{5750}
			\notag\\
			&=
			\frac{11482}{37375}
			=0.3072107\ldots
			\label{eq:76-fixed-gap}
		\end{align}
		or equivalently,
		\begin{equation}
			\big|k_{01}(1)-k_{01}\Delta V_0(1)\big|
			\geq
			\frac{11482}{37375}\,k_{01}.
			\label{eq:76-restriction-gap}
		\end{equation}
		Since $\alpha_n\to0$, for all sufficiently large $n$ we have $\alpha_n<\frac{11482}{37375}$.  For every such $n$, \eqref{eq:76-restriction-gap} violates the restriction
		\eqref{eq:Bal-restriction} with $(a,b,y)=(0,1,1)$, i.e., $\bx\notin\cC_n(\by,\bz)$ for every $\by\in\T_V^n(\bx)$ and every $\bz$ satisfying
		$\wt Q_n(\bz|\by)>0$.  Every indicator on the left-hand side of
		\eqref{eq:Bal-76} is therefore zero, proving
		\eqref{eq:Bal-76-zero}.
	\end{proof}
    
This proof reveals a major difference between the $(\bx,\by)$-dependent selector $\wh Q$ and the $\by$-only one $\wt Q$.  The definition of $\wh Q$ ensures that $\wh Q_{1,y}(1)=0$ whenever $y\in\cY_0^+$, so under $\wh Q$ the ``cross-count'' $k_{01}$ is zero for the transmitted codeword itself (see \eqref{eq:kab}), and \eqref{eq:Bal-restriction} holds trivially for $(a,b) = (0,1)$.  This is no longer the case after replacing $\wh Q_{x,y}$ by $\wt Q_y=\wh Q_{x(y),y}$.  In the above example, every coordinate with $y=2$ is selected regardless of the actual transmitted symbol, and the imbalance between $V_1(2)=0.5749$ and $V_1(1)=0.0001$ leads to a large discrepancy from the target distribution $\Delta V_0$.  Notably, this is not merely a concentration issue, as the probability under consideration eventually becomes exactly zero.

\subsection{Failure of \cite[Proposition 6.5]{Balakirsky1995} Itself} \label{sec:prop65-superposition}

The preceding subsection disproves two specific steps in the proof, but these failures do not in themselves disprove the asymptotic statement of \cite[Proposition 6.5]{Balakirsky1995}, claiming that one can set $\alpha_n$ satisfying \eqref{eq:alpha-65} to obtain
\begin{equation}
    \Upsilon\big(\{\Pe^{(n)}(V_n)\}\big) ~\substack{\textrm{claimed} \\ =}~ \Upsilon\big(\{\wh\Pe^{(n)}(V_n)\}\big),
    \label{eq:prop65-claim}
\end{equation}
where $V_n$ is the combinatorial channel in \eqref{eq:combinatorial-channel}, $\Upsilon(\cdot)$ is the asymptotic notation in \eqref{eq:upsilon}, and 
\begin{equation}
    \wh\Pe^{(n)}(V_n)
    =\max_{\bx\in\cC_n}\wh P_{\mathrm e}^{(n)}(\bx,V_n)
    \label{eq:Phat-max}
\end{equation}
with $\wh P_{\mathrm e}^{(n)}(\bx,V_n)$ being the codeword-dependent restricted error probability defined in \eqref{eq:Phat}.  In this subsection, we show that such a claim is also false:

\begin{lemma} \label{lem:upsilon_false}
	For the choices of $P$, $(W,q)$, and $V$ in \eqref{eq:P-counterexample}, \eqref{eq:counterexample-Wq}, and \eqref{eq:test-channel-V} respectively, there exists a  sequence of $P$-constant-composition codes (indexed by $n$) for which the identity in \eqref{eq:prop65-claim} is false for any sequence $\{\alpha_n\}$ satisfying \eqref{eq:alpha-65}. 
\end{lemma}

Specifically, we will establish a scenario under which the superposition coding construction of Section~\ref{sec:summary_ce} violates this claim, due to $\Pe^{(n)}(V_n)$ being bounded away from zero (by \eqref{eq:Fano-exact-V}) while the restricted error probability $\wh\Pe^{(n)}(V_n)$ approaches zero.  
When studying $\wh\Pe^{(n)}(V_n)$, as is standard in random coding, we will initially focus on the ensemble average:
\begin{equation}
	\wh P_{\mathrm{e,ens}}^{(n)}(\bx,V_n) = \E\big[ \wh P_{\mathrm e}^{(n)}(\bx,V_n) \big], \label{eq:ensemble-avg}
\end{equation}
where the average is over all non-transmitted (competitor) codewords in the random superposition coding ensemble.  The upper bounds will be independent of the transmitted codeword $\bx$ and will thus still hold after averaging over it.  A standard expurgation argument will then convert from average error to maximal error as used in \eqref{eq:Phat-max}.

The ensemble-average error probability will be decomposed into two types:
\begin{itemize}
    \item {\bf Same-cloud errors:} There exists a satellite in the same cloud as the transmitted codeword that induces a higher or equal decoding metric (and survives the restriction \eqref{eq:Bal-restriction});
    \item {\bf Wrong-cloud errors:} There exists a codeword from a different cloud to the transmitted codeword that induces a higher or equal decoding metric  (and survives the restriction \eqref{eq:Bal-restriction}).
\end{itemize}
Most of the effort will be in understanding the former, with the latter being deferred to Section \ref{sec:wrong-cloud}.  

\subsubsection{Useful technical definitions}

Define an augmented single-letter channel $S$ from $\cX$ to $\cY\times\{0,1\}$ by
\begin{align}
    S(y,0|x) & =\min\{V_x(y),W_x(y)\},
    \label{eq:S-zero}\\
    S(y,1|x) & =[V_x(y)-W_x(y)]^+.
    \label{eq:S-one}
\end{align}
For each $(x,y)$, checking the cases $V_x(y) \ge W_x(y)$ and $V_x(y) \le W_x(y)$ separately gives
\begin{equation}
    S(y,0|x)+S(y,1|x)=V_x(y),
    \label{eq:S-Y-marginal}
\end{equation}
which implies that $S$ is a valid conditional distribution, and its $Y$-marginal (given $x$) is exactly $V$.  The significance of $S$ is that the distribution $\wh Q_{x,y}$ defined in \eqref{eq:hatQ-one}--\eqref{eq:hatQ-zero} can be written as
\begin{equation}
    \wh Q_{x,y}(z)=\frac{S(y,z|x)}{V_x(y)}
    \label{eq:S-hatQ}
\end{equation}
whenever $V_x(y)>0$.  Consequently, if $S_n$ denotes the combinatorial channel that is uniform on the exact conditional type class $\T_S^n(\bx)$, then
\begin{equation}
    S_n(\by,\bz|\bx)
    =V_n(\by|\bx)\wh{Q}_n(\bz|\bx,\by).
    \label{eq:Sn-factorization}
\end{equation}
Hence, the pair $(\by,\bz)$ appearing in the definition of $\wh P_{\mathrm e}^{(n)}(\bx,V_n)$ in \eqref{eq:Phat} is simply the output of the combinatorial channel $S_n$. 

For the channels $W$ and $V$ in \eqref{eq:counterexample-Wq} and \eqref{eq:test-channel-V}, a direct calculation gives that the sets $\cY_x$ from \eqref{eq:Y-plus} are
\begin{equation}
    \cY_0^+=\{1,2\}, \qquad \cY_1^+=\{0\},
    \label{eq:Y-plus-example}
\end{equation}
and the distribution $\Delta V_0$ from \eqref{eq:Delta-V} is
\begin{equation}
    \Delta V_0(1)=\frac{999}{3250},
    \qquad
    \Delta V_0(2)=\frac{2251}{3250},
    \qquad
    \Delta V_1(0)=1.
    \label{eq:Delta-example}
\end{equation}
The $a=1$ part of the restriction in \eqref{eq:Bal-restriction} is therefore trivial (i.e., the left-hand side is identically zero).  The only non-trivial condition concerns the relative frequencies of output symbols $1$ and $2$ among selected coordinates associated with $a=0$.

Recall that in Section \ref{sec:summary_ce} the superposition code is applied over \emph{pairs} of symbols.  Accordingly, we group the binary length-$n$ block into $m = n/2$ adjacent pairs, so that the original block length is $n=2m$.  For one pair, write
\begin{equation}
    A=(U,X^\dagger,Y^\dagger,Z^\dagger),
    \label{eq:A-definition}
\end{equation}
where the $(\cdot)^{\dagger}$ notation refers to considering a pair of each associated symbol:
\begin{equation}
	 X^\dagger=(X_1,X_2), \quad Y^\dagger=(Y_1,Y_2), \quad Z^\dagger=(Z_1,Z_2). \
\end{equation}
Let $\overline{\bX}$ denote a (random) competing codeword from the same cloud as $\bx$, and let $\overline X^\dagger=(\overline X_1,\overline X_2)$ 
denote a generic pair associated with that codeword.  For $b\in\{0,1\}$, define
\begin{align}
    g_b(A,\overline X^\dagger)
    & =\sum_{\ell=1}^2 \1\{X_\ell=0,Z_\ell=1,\overline X_\ell=b\}
    \left( \1\{Y_\ell=1\} -\frac{999}{3250} \1\{Y_\ell\in\{1,2\}\} \right).
    \label{eq:gb-definition}
\end{align}
The significance of this definition is seen via the following lemma relating it to the restriction condition \eqref{eq:Bal-restriction}.  Here and subsequently, notation such as $A_i$ and $\overline X_i^\dagger$ represents the $i$-th pair among $m$ pairs in the sequence.

\begin{lemma} \label{lem:gb}
	For each $b\in\{0,1\}$ and each transmitted codeword $\bx$ and competitor codeword $\overline{\bX}$, every realization of $(\by,\bz)$ in the support of the exact-type channel $S_{2m}(\by,\bz | \bx)$ satisfies
	\begin{equation}
		\sum_{i=1}^m g_b(A_i,\overline X_i^\dagger)
		= k_{0b}(1)-\frac{999}{3250}k_{0b},
		\label{eq:gb-count-interpretation}
	\end{equation}
	which corresponds to the quantity inside $|\cdot|$ in \eqref{eq:Bal-restriction} with $a=0$ and $y=1$.  
	Moreover, the $y=2$ part of the $a=0$ restriction in \eqref{eq:Bal-restriction} is equivalent to the $y=1$ part, so \eqref{eq:gb-count-interpretation} captures the full restriction for $a=0$ (with the case $a=1$ being trivial as noted above).
\end{lemma}
\begin{proof}
	See Appendix \ref{app:pf-gb}
\end{proof}

Next, define the empirical joint distribution across $m$ pairs by
\begin{equation}
    \widehat P_{A\overline X^\dagger}(a,\bar x^\dagger)
    =\frac{1}{m}\sum_{i=1}^m
    \1\{(A_i,\overline X_i^\dagger)=(a,\bar x^\dagger)\}.
    \label{eq:pair-empirical-distribution}
\end{equation}
We claim that for any fixed $\alpha^* > 0$, every competitor satisfying the restriction in Definition~\ref{def:restriction} with parameter $\alpha_{2m}\to0$ necessarily obeys the following whenever $m$ is large enough:
\begin{equation}
    \Big|\E_{\widehat P}[g_b(A,\overline X^\dagger)]\Big|
    \leq2\alpha^*,
    \qquad b\in\{0,1\},
    \label{eq:gb-tolerance}
\end{equation}
where $\E_{\widehat P}$ denotes expectation under \eqref{eq:pair-empirical-distribution}.  Indeed, \eqref{eq:gb-count-interpretation} and \eqref{eq:Bal-restriction} give $\big|\E_{\widehat P}[g_b(A,\overline X^\dagger)]\big| \leq \frac{\alpha_{2m}k_{0b}}{m} \leq 2\alpha_{2m}$, with the last step using $k_{0b}\leq n=2m$.  Since $\alpha_{2m}\to0$, we also have $\alpha_{2m}\leq\alpha^*$ for all sufficiently large $m$, yielding \eqref{eq:gb-tolerance}.  Hence, replacing the full restriction \eqref{eq:Bal-restriction} by \eqref{eq:gb-tolerance} is a harmless relaxation that can only increase the restricted same-cloud error probability that we are studying (\emph{cf.}, \eqref{eq:Phat}, \eqref{eq:ensemble-avg}, and its decomposition stated thereafter).

Define the ``difference of two-letter log-metrics'' as
\begin{equation}
    \psi(A,\overline X^\dagger)
    =\sum_{\ell=1}^2
    \Big(\log{q(\overline X_\ell,Y_\ell)}
    - \log {q(X_\ell,Y_\ell)} \Big),
    \label{eq:metric-increment}
\end{equation}
and observe that a competitor that causes a decoding error must satisfy
\begin{equation}
    \E_{\widehat P}[\psi(A,\overline X^\dagger)]\geq0.
    \label{eq:metric-type-condition}
\end{equation}
Moreover, for $(x,y,z)\in\cX\times\cY\times\{0,1\}$, let
\begin{equation}
    \phi_{xyz}(A)
    =\sum_{\ell=1}^2
    \1\{X_\ell=x,Y_\ell=y,Z_\ell=z\},
    \label{eq:phi-definition}
\end{equation}
and define
\begin{equation}
    r_{xyz}=2P(x)S(y,z|x).
    \label{eq:rxyz}
\end{equation}
Because every codeword has composition $P$ and $S_{2m}(\by,\bz|\bx)$ fixes the exact conditional type of $(Y,Z)$ given $X$, every empirical distribution arising from \eqref{eq:Sn-factorization} satisfies
\begin{equation}
    \E_{\widehat P}[\phi_{xyz}(A)]=r_{xyz}.
    \label{eq:phi-constraint}
\end{equation}

\subsubsection{Method of types analysis}

Introduce the reference distribution
\begin{equation}
    B(a,\bar x^\dagger)
    =Q_{UX^2}(u,x^\dagger)
    \Big( \prod_{\ell=1}^2 S(y_\ell,z_\ell|x_\ell) \Big)
    Q_{X^2|U}(\bar x^\dagger|u),
    \label{eq:B-reference}
\end{equation}
where $a=(u,x^\dagger,y^\dagger,z^\dagger)$.  For the constant choice of $\alpha^* > 0$ introduced in \eqref{eq:gb-tolerance}, let $\cF_{\alpha^*}$ be the set of joint distributions $\wt P_{A\overline X^\dagger}$ satisfying
\begin{align}
    \wt P_{UX^\dagger}
    & =Q_{UX^2},
    \label{eq:F-transmitted-marginal}\\
    \wt P_{U\overline X^\dagger}
    & =Q_{UX^2},
    \label{eq:F-competitor-marginal}\\
    \E_{\wt P}[\phi_{xyz}(A)]
    & =r_{xyz},
    \qquad \text{for all }(x,y,z),
    \label{eq:F-exact-S}\\
    \E_{\wt P}[\psi(A,\overline X^\dagger)]
    & \geq0,
    \label{eq:F-metric}\\
    |\E_{\wt P}[g_b(A,\overline X^\dagger)]|
    & \leq2\alpha^*,
    \qquad b\in\{0,1\}.
    \label{eq:F-restriction}
\end{align}
Moreover, define
\begin{equation}
    \cE_{\alpha^*}
    =\min_{\wt P\in\cF_{\alpha^*}}
    D(\wt P_{A\overline X^\dagger}\|B).
    \label{eq:Ealpha-definition}
\end{equation}
The following analysis based on standard method of types arguments \cite[Ch.~2]{CsiszarKorner1981} gives an exponential upper bound on the probability of a fixed competitor codeword surviving the restriction and being favored over the true codeword by the maximum-metric decoder.

\begin{lemma} \label{lem:types}
	Fix $(\boldsymbol U,\boldsymbol X^\dagger)$ of joint type $Q_{UX^2}$.  Draw $(\boldsymbol Y^\dagger,\boldsymbol Z^\dagger)$ (given $\bX^\dagger$) according to $S_{2m}$ in \eqref{eq:Sn-factorization} and independently (given $\boldsymbol U$) draw one competing same-cloud codeword $\overline{\boldsymbol X}^\dagger$ uniformly from its conditional type class corresponding to $Q_{X^2|U}$.  Let $\cE_{1,m}$ be the event that this competitor satisfies both \eqref{eq:Bal-restriction} and
	\begin{equation}
		q^{2m}(\overline{\boldsymbol X}^\dagger,\boldsymbol Y^\dagger)
		\geq
		q^{2m}(\boldsymbol X^\dagger,\boldsymbol Y^\dagger).
		\label{eq:E1m-definition}
	\end{equation}
	Then
	\begin{equation}
		\Pp[\cE_{1,m}\mid \boldsymbol U,\boldsymbol X^\dagger]
		\leq
		\exp\big(-m\cE_{\alpha^*}+O(\log m)\big).
		\label{eq:one-satellite-bound}
	\end{equation}
\end{lemma}
\begin{proof}
	The proof is based on a standard method of types analysis, so is deferred to Appendix \ref{app:pf-types}.
\end{proof}

Since there are at most $e^{mR_1}$ incorrect satellites in the transmitted cloud, a union bound gives
\begin{equation}
    \Pp[\text{restricted same-cloud error}]
    \leq
    \exp\big(-m(\cE_{\alpha^*}-R_1)+O(\log m)\big).
    \label{eq:same-cloud-union}
\end{equation}
To establish that this approaches zero, it suffices to show that $\cE_{\alpha^*}$ is strictly higher than $R_1$ for all sufficiently large $m$.  The following subsection establishes that this is true.

\subsubsection{Exponential decay via a dual certificate} \label{sec:dual-cert}

The optimization problem in \eqref{eq:Ealpha-definition} is convex and can readily be solved using standard solvers.  We give such an approach in the code \cite{ScarlettCode2026}, but in the proof of the following lemma we instead form an explicit lower bound and a ``dual certificate'' for concreteness and to avoid strictly relying on numerical optimization (instead only using a single numerical evaluation).

\begin{lemma} \label{lem:dual-cert}
	Under choices of $P$, $(W,q)$, $V$, and $Q_{UX^2}$ in \eqref{eq:P-counterexample}, \eqref{eq:counterexample-Wq}, \eqref{eq:test-channel-V}, and \eqref{eq:QU}--\eqref{eq:QXU1} respectively, the optimization problem in \eqref{eq:Ealpha-definition} satisfies
	\begin{equation}
		\cE_{\alpha^*}
		\geq 0.35870725 - 60.444\alpha^* \label{eq:Ealpha-certificate}
	\end{equation}
	for every $\alpha^*>0$.
\end{lemma}
\begin{proof}
	See Appendix \ref{app:pf-dual-cert}.
\end{proof}

This result implies that for sufficiently small $\alpha^* > 0$, it holds that
\begin{equation}
    \cE_{\alpha^*}
    \geq0.3587>R_1=0.240,
    \label{eq:Ealpha-above-R1}
\end{equation}
recalling this choice of $R_1$ from \eqref{eq:backed-off-rates}.  This implies that \eqref{eq:same-cloud-union} tends to zero exponentially fast.

\subsubsection{Wrong-cloud errors} \label{sec:wrong-cloud}

The preceding calculation handles competitors belonging to the same cloud as the transmitted codeword.  It remains to control competitors associated with a different cloud center.  Since Balakirsky's restriction only deletes competitors, it is enough to bound the ordinary wrong-cloud error probability for superposition coding over $V^2$.  We use the cloud-rate dual formula \eqref{eq:SC-dual-B0}, taking
\begin{equation}
    \rho=0.76, \qquad s=9.34,
    \label{eq:wrong-cloud-rho-s}
\end{equation}
and the following auxiliary costs:
\begin{equation}
    \begin{array}{c|rrrr}
        a(U,X^{\dagger}) &X^{\dagger} = 00&X^{\dagger} = 01&X^{\dagger} = 10&X^{\dagger} = 11\\ \hline
        U=0&0&-0.767&-0.767&0\\
        U=1&0&0&0&-1.341
    \end{array}.
    \label{eq:wrong-cloud-costs}
\end{equation}
At $R_1=0.240$, direct substitution into \eqref{eq:SC-dual-B0} (with the channel $W$ therein replaced by $V^2$) gives the following (see \cite{ScarlettCode2026} for numerical verification):
\begin{equation}
    R^{\rm SC}_0(0.76,9.34,a;0.240)
    =0.0360110002\ldots>R_0=0.035.
    \label{eq:wrong-cloud-certificate}
\end{equation}
Hence, the wrong-cloud error under $m$ independent uses of $V^2$ (i.e., $n$ independent uses of $V$) tends to zero exponentially fast.  Here we are interested in the combinatorial channel $V_{n}$, but the same conclusion still holds because conditioning the memoryless channel $V^n$ on its exact conditional type costs only a polynomial factor (similarly to \eqref{eq:type-prob-lower}). Moreover, imposing Balakirsky's restriction can only further reduce this error probability, and as a result, the restricted wrong-cloud error probability also tends to zero.

\subsubsection{Completion of the argument}

Combining the findings from Sections~\ref{sec:dual-cert} and \ref{sec:wrong-cloud}, we have established that the ensemble-average restricted error under Balakirsky's exact-type
$V_{2m}$ construction tends to zero exponentially fast.  Moreover, we have shown this
using a fixed $\alpha^*>0$ in \eqref{eq:gb-tolerance}, independent of the
particular sequence $\{\alpha_n\}$, with the actual restriction implying
\eqref{eq:gb-tolerance} for all sufficiently large $n$ whenever  \eqref{eq:alpha-65} holds.  It follows that there exists a deterministic
sequence of superposition codebooks, chosen independently of
$\{\alpha_n\}$, whose average error under this fixed relaxed restriction tends to zero.  

Expurgating a fixed fraction of messages then yields a single $P$-constant-composition code sequence with maximal error under the relaxed restriction tending to zero, without changing the asymptotic rate.  Since the original restriction is stronger, the maximal restricted error also tends to zero for every admissible sequence $\{\alpha_n\}$.

On the other hand, the ordinary exact-type $V_{2m}$ error for the resulting
code sequence is bounded away from zero by \eqref{eq:Fano-exact-V}.  Combining the preceding two findings gives
\begin{equation}
    \Upsilon\big(\{\Pe^{(n)}(V_n)\}\big)=1,
    \qquad
    \Upsilon\big(\{\wh\Pe^{(n)}(V_n)\}\big)=0,
    \label{eq:prop65-contradiction}
\end{equation}
which contradicts the claim in \eqref{eq:prop65-claim} and establishes Lemma \ref{lem:upsilon_false}.

\subsection{Other Potential Gaps in the Analysis} \label{sec:other}

Here we describe some further issues in the analysis in \cite{Balakirsky1995}, including some steps that are false as stated.  These issues are separate from the failures in Sections~\ref{sec:perm-false}--\ref{sec:prop65-superposition}; they may be more prone to local repair, and some of them are alleviated under additional restrictions on the rate or channel.

\subsubsection{Issues in the proof of \cite[Proposition~6.7]{Balakirsky1995}}
\label{sec:Prop6.7}

The statement of \cite[Proposition 6.7]{Balakirsky1995} is that, for some sequence $\alpha_n$ satisfying
\begin{equation}
	\alpha_n\to0,
	\qquad
	\alpha_n\sqrt{n}\to\infty,
	\label{eq:alpha-condition}
\end{equation}
the restricted decoding error probability under the combinatorial channel $W_n$ remains non-vanishing whenever the corresponding restricted decoding error probability under $V_n$ does, i.e., 
\begin{equation}
	\Upsilon\big(\{\wh P_e^{(n)}(V_n)\}\big) \leq
	\Upsilon\big(\{\wh P_e^{(n)}(W_n)\}\big), \label{eq:67-statement}
\end{equation}
where $\wh P_e^{(n)}(V_n)$ is defined in \eqref{eq:Phat-max} (via \eqref{eq:Phat}) and $\wh P_e^{(n)}(W_n)$ is defined analogously
via $\Delta W$ whose definition matches \eqref{eq:Delta-V} but with the roles of $V$ and $W$ interchanged.  

One of the main steps in the proof is a concentration argument relating to the restriction condition \eqref{eq:Bal-restriction} \cite[second step after Eq.~(84)]{Balakirsky1995}.  The relevant empirical counts are formed within competitor-dependent sub-blocks of size $k_{ab}$, and thus exhibit typical fluctuations of order $\sqrt{k_{ab}}$.  Hence, one would need $\alpha_n\sqrt{k_{ab}}\to\infty$ for the tolerance in \eqref{eq:Bal-restriction} to dominate these fluctuations.  The argument treats the conditions in \eqref{eq:alpha-condition} as sufficient for this purpose, and while this is true when $k_{ab}=\Theta(n)$, it can fail when $k_{ab}=o(n)$.

For an arbitrary deterministic codebook, there is no lower bound forcing the relevant positive $k_{ab}$ values to grow with $n$; for instance, codeword pairs differing in only a few coordinates can even yield $k_{ab}=1$.  Hence, the concentration argument does not appear to be justified as stated.  A possible repair would be to first expurgate the codebook (without changing the asymptotic rate) to enforce sufficiently large pairwise distances, and then prove that the selector retains a large enough fraction of the differing coordinates.  However, this additional step appears to be non-trivial, since the $k_{ab}$ values here are defined with respect to an error-causing competitor in \eqref{eq:Phat}, which may itself depend on the realized output and selector sequence.  Hence, fixed-pair concentration may be insufficient.

Another issue in the proof of  \cite[Proposition~6.7]{Balakirsky1995} is that it temporarily assumes $d(P,V)\leq d(P,W)-2\alpha_n d_{\max}$, and then states that $\alpha_n$ can be set to zero when formulating the final $\Upsilon$ result.  This does not by itself cover the boundary case $d(P,V)=d(P,W)$, since the same sequence must also satisfy \eqref{eq:alpha-condition}.  This is a seemingly minor issue that may be resolved by a perturbation or continuity argument.

We note that the overall statement of  \cite[Proposition 6.7]{Balakirsky1995} may still plausibly be true. %, particularly in the regime $R > I(P,V)$ relevant to the converse argument.

\subsubsection{Combinatorial approximation lemma \cite[Lemma~4.4]{Balakirsky1995} and its proof} \label{eq:perturbed}

The combinatorial approximation lemma \cite[Lemma~4.4]{Balakirsky1995} concerns the first and last steps in the overview \eqref{eq:Bal-steps}, for channels $V$ and $W$ respectively.  We highlight the following potential issues concerning the proof and the lemma itself.

{\bf The nearest-output comparison in \cite[Eq.~(42)]{Balakirsky1995}.} This step is relevant to both $V$ and $W$, but we use the notation $V$ for concreteness.  
	Fix $\delta_n > 0$, let $[V]$ be the set of conditional types within $\delta_n$ of $V$ in maximum absolute difference (i.e., $\ell_{\infty}$-norm), and let $\T_{[V]}^n(\bx) = \bigcup_{V' \in [V]} \T_{V'}^n(\bx)$.  
	In this part of the analysis, an output
	$\by\in\T_{[V]}^n(\bx)$ is mapped to a nearest output
	$\underline\by\in\T_{\underline V}^n(\bx)$ in Hamming distance, where $\underline V$ is the ``favorable
	perturbation'' in \cite[Eq.~(36)]{Balakirsky1995}: $\underline V_0(y)=V_0(y), \forall y$, and
	\begin{equation}
		\underline V_1(y)=
		\begin{cases}
			V_1(y)+(|\cY|-1)\delta_n, & y=0,\\
			V_1(y)-\delta_n, & y\neq 0
		\end{cases}
		\label{eq:V-underline}
	\end{equation}
	for some sequence $\{\delta_n\}$ approaching zero. The first inequality in
	\cite[Eq.~(42)]{Balakirsky1995} claims the following:
	\begin{equation}
		\text{correct decoding under $\by$} \stackrel{\mathrm{claimed}}{\implies} \text{correct decoding under $\by'$}. \label{eq:claimed-42}
	\end{equation}
	This turns out to be false in general since, although $\underline V_0=V_0$, the empirical distribution induced by $\by$ may differ slightly from $V_0$ due to fluctuations in the $n$ memoryless uses of $V$.  The required ``correction'' may reduce a competitor's distortion enough to induce an error.
	
	For example, take $n=16$, $\mathcal{Y} = \{0,1\}$, $P=(1/2,1/2)$, $V_0=V_1=(1/2,1/2)$,
	$\delta_{n}=1/4$, $d_0=(0,0)$, and $d_1=(0,1)$.  Then
	$\underline V_0=(1/2,1/2)$ and $\underline V_1=(3/4,1/4)$.  Consider the
	following two-codeword code and two outputs:
	\begin{equation}
		\begin{aligned}
			\bx&=00000000\,11111111,
			&\bxbar&=10000000\,01111111,\\
			\by&=11111000\,00001111,
			&\underline\by&=01111000\,00000011.
		\end{aligned}
		\label{eq:nearest-output-example}
	\end{equation}
	The conditional empirical distributions induced by $\by$ are $(3/8,5/8)$ (given $x=0$) and $(1/2,1/2)$ (given $x=1$), so
	$\by\in\T_{[V]}^{n}(\bx)$.  The output $\underline\by$ has exact
	conditional type $\underline V$, and its distance from $\by$ is the
	minimum possible: One change is needed in the $x=0$ sub-block and two in the $x=1$ sub-block.  Nevertheless, we have
	\begin{gather}
			d^{n}(\bx,\by)=4<5=d^{n}(\bxbar,\by),\\
			d^{n}(\bx,\underline\by)=2=d^{n}(\bxbar,\underline\by). 
	\end{gather}
	Hence, the nearest output creates a tie, which contradicts \eqref{eq:claimed-42} since ties are counted as errors.
	
	While the above example concerns finite $n$, the same issue can persist for arbitrarily large $n$.  For even $s$, take $n=2s^4$ and $\delta_n=1/s$, so that $\delta_n\to0$ and $\delta_n\sqrt n\to\infty$.  Let $\bx$ again have all 0s in the first half and all 1s in the second half, and choose $\by$ to have $s^4/2+1$ output-1 symbols in the first half and $s^4/2$ of them in the second half.  Choose the two differing codeword coordinates at an output-1 position in the first half and an output-0 position in the second half.  A nearest output in $\T_{\underline V}^n(\bx)$ can then flip the former symbol together with $s^4\delta_n=s^3$ output-1 symbols in the second half, again creating a tie.
	
	It is plausible that \cite[Eq.~(42)]{Balakirsky1995} could be suitably modified or restricted to mild conditions that preclude the above scenarios or similar counter-examples.

{\bf Comparison of three channels in \cite[Eq.~(39)]{Balakirsky1995}.} 
Following \cite[Proposition~5.1 and Eqs.~(37)--(40)]{Balakirsky1995}, Balakirsky considers three channels that are ``nearby'' with closeness parameter $\delta_n \to 0$ in the same sense as $V$ and $\underline{V}$ from \eqref{eq:V-underline}.  He invokes the ordinary converse for these channels and concludes (in \cite[Eq.~(39)]{Balakirsky1995}) that the corresponding $\Upsilon$ behaviors (see \eqref{eq:upsilon}) of the error probability are equal for the given code sequence.  This indeed holds when the fixed code rate is strictly above $I(P,K)$ (where $K \in \{V,W\}$ is the channel under consideration), since it is then eventually above the mutual information of all three channels.  However, the combinatorial approximation lemma is not restricted to this regime, and it is less clear whether the corresponding asymptotic ($\Upsilon$) equality is guaranteed for arbitrary codes at rates smaller than or equal to $I(P,K)$, since this may require sufficient robustness to perturbations that depend on $n$ through $\delta_n$.

{\bf Generality of the lemma itself.} While the above discussion concerns a specific proof step, we also note that the combinatorial approximation lemma itself may be less general than stated (regardless of the proof strategy).  For instance, at the capacity boundary with $\exp(nC + O(1))$ messages, the matched decoding error probability is at least $1/2 - o(1)$ for positive-dispersion channels by central limit theorem arguments \cite{PolyanskiyPoorVerdu2010}, but it can be shown that the combinatorial channel error probability may still approach zero.  However, this is a highly degenerate example, and it does not rule out the possibility of a version of the lemma that holds in the required generality.

{\bf Discussion.} The combinatorial approximation lemma is applied to both $V$ and $W$, and the regime of primary interest is $I(P,V) < R < I(P,W)$.  This restriction appears to alleviate the concern regarding \cite[Eq.~(39)]{Balakirsky1995} on the $V$ side, but not on the $W$ side.  Moreover, a suitable amendment to \cite[Eq.~(42)]{Balakirsky1995} would still be required.  Overall, we expect that the statement of \cite[Lemma~4.4]{Balakirsky1995} is less general than claimed, but it may still be true in the generality needed for its application.

\section{Conclusion}

While the numerical counter-example based on superposition coding in \cite{ScarlettEtAl2015} successfully refuted Balakirsky's converse \cite{Balakirsky1995}, it left open the question of which steps in his analysis may be flawed.  In this paper, we have established the following:
\begin{enumerate}[label=(\roman*)]
	\item The permutation lemma is itself false, as we established in Section \ref{sec:perm-false} using the superposition coding counter-example from \cite{ScarlettEtAl2015}.
    \item Two specific steps in the proof of \cite[Proposition 6.5]{Balakirsky1995} based on selector sequences and restricted codes are false, as we established in Section~\ref{sec:false-steps}.
	\item The conclusion of \cite[Proposition 6.5]{Balakirsky1995} is itself false for the same superposition coding counter-example as \cite{ScarlettEtAl2015}, as we established in Section \ref{sec:prop65-superposition}.
\end{enumerate}
We also raised other potential issues in Section \ref{sec:other}, though these may potentially be locally repairable.  Moreover, we do not claim that the (potential) flaws identified in this paper are necessarily exhaustive.

\appendix

\section{A Correction to  \cite[Lemma 1]{ScarlettEtAl2015}} \label{sec:our_correction}

In the work providing the numerical counter-example to the binary-input converse, \cite[Lemma 1]{ScarlettEtAl2015} gives a continuity bound for the binary entropy function.  The lemma states that for two distributions $Q,Q'$ on $\{0,1\}$ satisfying $|Q(0)-Q'(0)|\leq\delta$ and $\min_xQ'(x)\geq q_{\min}'$, it holds that
\begin{equation}
    |H(Q)-H(Q')|
    \stackrel{\textrm{claimed}}{\leq} \delta\log\frac{1-q_{\min}'}{q_{\min}'}.
    \label{eq:old-continuity}
\end{equation}
As stated, this is not valid when $Q$ moves from $Q'$ toward the boundary of the simplex.  For example, under the choices
\begin{equation}
    Q'=(0.2,0.8),
    \qquad
    Q=(0.1,0.9),
    \qquad
    \delta=0.1,
    \label{eq:entropy-counterexample}
\end{equation}
the left-hand side of \eqref{eq:old-continuity} is roughly $0.17532$, whereas the right-hand side is roughly $0.13863$. 

A corrected version is obtained by defining
$q_{\min}=\min_{x\in\{0,1\}}\min\{Q(x),Q'(x)\}$, and replacing $q'_{\min}$ in \eqref{eq:old-continuity} by this value:
\begin{equation}
    |H(Q)-H(Q')|
    \leq \delta\log\frac{1-q_{\min}}{q_{\min}}.
    \label{eq:correct-continuity}
\end{equation}
Indeed, the derivative of binary entropy is $\log((1-t)/t)$, and every point on the line segment between $Q$ and $Q'$ has both coordinates at least $q_{\min}$.

Fortunately, this correction does not affect the counter-example.  In the notation of \cite[Lemma~1]{ScarlettEtAl2015}, we have $q_{\min}\geq q_{\min}'-\delta$, so \eqref{eq:correct-continuity} is in turn upper-bounded by
$\delta\log\frac{1-q_{\min}'+\delta}{q_{\min}'-\delta}$ (assuming $\delta \in (0,q'_{\min})$).  For the two parameter choices used in the paper, the resulting total continuity penalties change as follows (see \cite{ScarlettCode2026} for numerical verification):
\begin{equation}
    \begin{array}{c|cc}
        q_{\min}' & \text{old value} & \text{corrected value}\\ \hline
        0.042 & 9.8153412\times10^{-5} & 9.8154034\times10^{-5}\\
        0.2   & 2.4260151\times10^{-5} & 2.4260308\times10^{-5}
    \end{array}.
    \label{eq:continuity-penalties}
\end{equation}
The rounded bounds used in the paper remain valid, and in particular the conclusion $\Clm\leq0.136900$ is unchanged.

\section{Proof of Lemma \ref{lem:wrong-step-75} (Counter-Example to \cite[Eq.~(75)]{Balakirsky1995})} \label{app:pf-75}

Let $\wc{P}_{\mathrm e}^{(n)}(\bx,V_n)$ denote the right-hand side of \eqref{eq:Bal-75}:
\begin{equation}
	\wc{P}_{\mathrm e}^{(n)}(\bx,V_n) := 
	\sum_{\by,\bz} V_n(\by|\bx)\wt{Q}_n(\bz|\by)
	\1\{\bx\in\cC_n(\by,\bz)\}
	\1\{d^n(\bx,\by)\geq D_n(\bx,\by|\bz)\}. \label{eq:Pe_tilde_def} 
\end{equation}
The notation $\wc{P}$ is not used in \cite{Balakirsky1995}, but it denotes precisely the second term in \cite[Eq.~(74)]{Balakirsky1995}.  Balakirsky defines $\wt{P}$ therein to be the sum of this $\wc{P}$ term and the probability that $\bx$ fails the restriction (thus counting it as an error), but we will not need $\wt{P}$ here.  We proceed in several steps.  

{\bf Specification of problem parameters.} 
Take $n=8$, $P=(1/2,1/2)$, $\mathcal{Y}=\{0,1,2,3\}$, and
\begin{equation}
	V_0=V_1=(1/4,1/4,1/4,1/4),
	\label{eq:finite-V}
\end{equation}
\begin{equation}
	W_0=(0,0,1/2,1/2), \qquad W_1=(1/2,1/2,0,0).
	\label{eq:finite-W}
\end{equation}
It follows readily from \eqref{eq:finite-V}--\eqref{eq:finite-W} that
$PV=PW=(1/4,1/4,1/4,1/4)$, and that $0 = I(P,V) < I(P,W) = \log 2$, so that the constraints \eqref{eq:bal-marginal}--\eqref{eq:bal-mutual-information} are satisfied.

We consider the all-zero distortion function, i.e., $d_x(y) = 0$ for all $(x,y)$, corresponding to $q(x,y) = 1$, and trivially ensuring that \eqref{eq:bal-distortion} holds.  While this choice is somewhat non-ideal in that it gives a mismatch capacity of zero, it is not ruled out by the assumptions in \cite{Balakirsky1995}, and it suffices for proving Lemma \ref{lem:wrong-step-75}.  We comment on some possible alternative choices as follows:
\begin{itemize}
	\item With only slightly extra effort, it can be shown that the following choice also serves as a counter-example for the same $(P,W,V)$:
	\begin{equation}
		d_0(y)=0, \qquad \big(d_1(0),d_1(1),d_1(2),d_1(3)\big)=(0,0,1,1).
		\label{eq:finite-distortion}
	\end{equation}
	This choice leads to a mismatch capacity of $\log 2$ nats/use (i.e., 1 bit/use), but it is still non-ideal since it violates the condition $d(P,V) \le d(P,W)$ from \eqref{eq:bal-distortion}, thus being inadmissible (though that condition was not used in the proof of \cite[Eq.~(75)]{Balakirsky1995}).
	\item The author's AI agent proposed a candidate counter-example that was claimed to be admissible and yield a positive mismatch capacity, but it used $n=10$, $|\mathcal{Y}| = 6$, and five codewords, thus being significantly more complicated.  The author judged that this downside significantly outweighs the benefit, so did not attempt to verify its validity.
\end{itemize}

%We use the distortion function (corresponding to $d_x(y) = -\log q(x,y)$) given by
%\begin{equation}
%	d_0(y)=0, \qquad \big(d_1(0),d_1(1),d_1(2),d_1(3)\big)=(0,0,1,1).
%	\label{eq:finite-distortion}
%\end{equation}
%	For example, \eqref{eq:finite-distortion} corresponds to
%	\begin{equation}
	%		q=\begin{bmatrix}1&1&1&1\\1&1&e^{-1}&e^{-1}\end{bmatrix}.
	%	\end{equation}
% While the all-zero distortion measure would also suffice for this proof, we prefer to use a non-trivial distortion function to ensure that a positive rate is at least attainable. In fact, it gives a mismatch capacity of $\log 2$ nats/use (i.e., 1 bit/use): For a $P$-constant-composition code, the transmitted codeword has distortion zero, while every distinct codeword has strictly positive distortion.

{\bf Simplifying the restriction conditions.} 
From \eqref{eq:finite-V}--\eqref{eq:finite-W} and the definition in
\eqref{eq:Y-plus}, we have
\begin{equation}
	\mathcal{Y}_0^+=\{0,1\}, \qquad \mathcal{Y}_1^+=\{2,3\},
	\label{eq:finite-Yplus}
\end{equation}
and hence
\begin{equation}
	x(y)=\begin{cases}0,&y\in\{0,1\},\\1,&y\in\{2,3\}.\end{cases}
	\label{eq:finite-xy}
\end{equation}
Using \eqref{eq:Delta-V}, we similarly obtain
\begin{equation}
	\Delta V_0=(1/2,1/2,0,0), \qquad \Delta V_1=(0,0,1/2,1/2).
	\label{eq:finite-Delta}
\end{equation}
Moreover, the $(x,y)$-dependent selector distribution in \eqref{eq:hatQ-one}--\eqref{eq:hatQ-zero} is deterministic in the sense that it assigns conditional probability 1 to a particular $z$ value depending on $(x,y)$:
\begin{equation}
	\big(\widehat Q_{0,y}(1)\big)_{y=0}^3=(1,1,0,0), \qquad
	\big(\widehat Q_{1,y}(1)\big)_{y=0}^3=(0,0,1,1).
	\label{eq:finite-Qhat}
\end{equation}
Combining this with \eqref{eq:finite-xy} and \eqref{eq:tildeQ-single}, we obtain the following for the $y$-dependent selector distribution:
\begin{equation}
	\widetilde Q_y(1)=\widehat Q_{x(y),y}(1)=1, \qquad y\in\mathcal{Y}.
	\label{eq:finite-Qtilde}
\end{equation}
In other words, the $y$-only selector always sets $z_j=1$.

Using  \eqref{eq:finite-Delta} and the definitions in \eqref{eq:kaby}--\eqref{eq:kab}, there are only two non-trivial cases in \eqref{eq:Bal-restriction}:
\begin{itemize}
	\item when $a = 0$ and $y \in \{0,1\}$, \eqref{eq:Bal-restriction} gives
	\begin{equation}
		\left|k_{0b}(0)-\frac{k_{0b}}{2}\right|\le\alpha_n k_{0b}, \qquad
		\left|k_{0b}(1)-\frac{k_{0b}}{2}\right|\le\alpha_n k_{0b}
		\label{eq:finite-restrict-0}
	\end{equation}
	\item when $a= 1$ and $y \in \{2,3\}$, \eqref{eq:Bal-restriction} gives
	\begin{equation}
		\left|k_{1b}(2)-\frac{k_{1b}}{2}\right|\le\alpha_n k_{1b}, \qquad
		\left|k_{1b}(3)-\frac{k_{1b}}{2}\right|\le\alpha_n k_{1b}
		\label{eq:finite-restrict-1}
	\end{equation}
\end{itemize}
since the other combinations of $(a,y)$ give a left-hand side of zero in \eqref{eq:Bal-restriction} so that the condition trivially holds.  Moreover, the relevant $k_{ab}$ values simplify as follows after removing terms that are zero:
\begin{equation}
	k_{0b}=k_{0b}(0)+k_{0b}(1), \qquad
	k_{1b}=k_{1b}(2)+k_{1b}(3), \qquad b\in\{0,1\}.
	\label{eq:finite-kab}
\end{equation}
The subsequent analysis holds for any $\alpha_n\in (0,1/2)$ (and $n=8$).

{\bf Specification of code and initial analysis.}
Consider the two-codeword constant-composition code
\begin{equation}
	\mathcal{C}_n=\{\bx,\bar{\bx}\}, \qquad
	\bx=00001111, \qquad \bar{\bx}=01101001,
	\label{eq:finite-code}
\end{equation}
and suppose that $\bx$ is transmitted.   
%Since the two codewords in
%\eqref{eq:finite-code} differ precisely in coordinates $2,3,6,7$, the distortion function in \eqref{eq:finite-distortion} gives
%\begin{equation}
%	d^n(\bar{\bx},\by)-d^n(\bx,\by)
%	=d_1(y_2)+d_1(y_3)-d_1(y_6)-d_1(y_7).
%	\label{eq:finite-dist-diff}
%\end{equation}
We describe the combinatorial channel law $V_n$, letting
$\boldsymbol{Y}\sim V_n(\cdot|\bx)$.  By \eqref{eq:finite-V} and the fact that $\bx=00001111$ is transmitted, we have that $(Y_1,\ldots,Y_4)$ is a uniformly random
permutation of $\mathcal{Y}$, and independently $(Y_5,\ldots,Y_8)$ is
another uniformly random permutation of $\mathcal{Y}$.  Hence,
$\boldsymbol{Y}$ is uniform over
\begin{equation}
	(4!)^2 = 24^2 = 576
	\label{eq:finite-total-outputs}
\end{equation}
possible output vectors.  Both selectors in
\eqref{eq:finite-Qhat} and \eqref{eq:finite-Qtilde} are deterministic
given $\boldsymbol{Y}$ (and given that $\bx$ is transmitted), so all probabilities below refer only to this  randomness in $\boldsymbol{Y}$.

For $h\in\{L,R\}$ and $y\in\mathcal{Y}$, let
$\pos_h(y)$ denote the unique position in the left ($L$) or right ($R$) 
length-four half at which output symbol $y$ occurs.

{\bf Analysis of the $(\bx,\by)$-dependent selector.}
We first characterize when the competitor $\bar{\bx}$ survives the restriction.

\begin{lemma} \label{lem:xy-selector}
	Under the deterministic $(\bx,\by)$-dependent selector in \eqref{eq:finite-Qhat}, the competitor $\bar{\bx}$ survives the restriction in Definition \ref{def:restriction} if and only if
	\begin{equation}
		\bar x_{\pos_L(0)}=\bar x_{\pos_L(1)}
		\qquad\text{and}\qquad
		\bar x_{\pos_R(2)}=\bar x_{\pos_R(3)}.
		\label{eq:finite-Qhat-survival}
	\end{equation}
\end{lemma}

\begin{proof}
	By \eqref{eq:finite-Qhat}, in the left half, the selector sets $z_j=1$
	precisely at the positions occupied by output symbols $0$ and $1$.
	Hence, and using $x(y)$ from \eqref{eq:finite-xy}, the relevant $k_{ab}(y)$ values (defined in \eqref{eq:kaby}) simplify as follows for each $b\in\{0,1\}$:
	\begin{equation}
		k_{0b}(0)=\mathbf{1}\{\bar x_{\pos_L(0)}=b\}, \qquad
		k_{0b}(1)=\mathbf{1}\{\bar x_{\pos_L(1)}=b\}.
		\label{eq:finite-k-first}
	\end{equation}
	If $\pos_L(0)$ and $\pos_L(1)$ correspond to different competitor bits, then one such bit is $0$ and one is $1$, giving (via \eqref{eq:finite-kab}) that $k_{0b}=k_{0b}(0)+k_{0b}(1)=1$ 
	for both values of $b\in\{0,1\}$, and making both left-hand sides in \eqref{eq:finite-restrict-0} equal $1/2$. Since $\alpha_n<1/2$, the restriction fails.  Conversely, if the two
	positions have the same competitor bit, then for one value of $b$ the
	two counts in \eqref{eq:finite-k-first} are both one, while for the
	other they are both zero, and \eqref{eq:finite-restrict-0} holds due to the left-hand sides being zero.
	Therefore, the $a=0$ part of the restriction holds if and only if $\bar x_{\pos_L(0)}=\bar x_{\pos_L(1)}$.
	
	Similarly, in the right half, \eqref{eq:finite-Qhat} sets $z_j=1$
	precisely at the positions occupied by symbols $2$ and $3$, and
	\begin{equation}
		k_{1b}(2)=\mathbf{1}\{\bar x_{\pos_R(2)}=b\}, \qquad
		k_{1b}(3)=\mathbf{1}\{\bar x_{\pos_R(3)}=b\}.
		\label{eq:finite-k-second}
	\end{equation}
	The same argument as above then shows that \eqref{eq:finite-restrict-1} holds if and
	only if $\bar x_{\pos_R(2)}=\bar x_{\pos_R(3)}$.  Combining these two conditions gives \eqref{eq:finite-Qhat-survival}.
\end{proof}

The competitor bits in the left half are
\begin{equation}
	(\bar x_1,\bar x_2,\bar x_3,\bar x_4)=(0,1,1,0).
	\label{eq:finite-first-bits}
\end{equation}
Among the $4! = 24$ permutations of the output symbols $(0,1,2,3)$ in this half, exactly
$2^3=8$ place $0$ and $1$ at positions having matching $\bar{x}$ values:
One factor of $2$ chooses whether this common $\bar{x}$ value is $0$ or $1$,
another factor of $2$ comes from assigning $0$ and $1$ to the corresponding two positions,
and the final factor of $2$ comes from assigning $2$ and $3$ to the remaining positions.

Similarly, in the right half,
\begin{equation}
	(\bar x_5,\bar x_6,\bar x_7,\bar x_8)=(1,0,0,1),
	\label{eq:finite-second-bits}
\end{equation}
and the same counting argument shows that exactly $8$ of the $24$
possible right-half permutations place output symbols $2$ and $3$ at
positions having matching competitor bits.

The two halves are independent under $V_n$, and $\bar{\bx}$ survives
the restriction for
\begin{equation}
	8\cdot 8=64
	\label{eq:finite-hat-survive}
\end{equation}
of the $24^2 = 576$ outputs in \eqref{eq:finite-total-outputs}.  Moreover, 
we claim that the transmitted word $\bx$ itself satisfies the restriction for every
output.  To see this, note that in the left half $\bx$ is all 0, and from \eqref{eq:finite-Qhat} the selector sets $z_j = 1$ precisely when $y_j \in \{0,1\}$, giving $k_{00}(0) = k_{00}(1) = 1$, and by a similar argument in the right half (with $\bx$ having all 1s), $k_{11}(2)=k_{11}(3)=1$, with the remaining relevant counts being zero.  Hence, the left-hand sides in \eqref{eq:finite-restrict-0}--\eqref{eq:finite-restrict-1} are all zero.

Under the all-zero distortion metric with ties counted as errors, all of the $64$ outputs corresponding to \eqref{eq:finite-hat-survive} cause a decoding error.  It follows that
\begin{equation}
	\widehat P_e^{(n)}(\bx,V_n)=\frac{64}{576}=\frac{1}{9}.
	\label{eq:finite-Pehat}
\end{equation}

%It remains to determine which of the $64$ outputs corresponding to
%\eqref{eq:finite-hat-survive} cause a decoding error.  Define
%\begin{equation}
%	A=d_1(Y_2)+d_1(Y_3), \qquad B=d_1(Y_6)+d_1(Y_7),
%	\label{eq:finite-AB}
%\end{equation}
%so that the distortion difference in \eqref{eq:finite-dist-diff} is
%$A-B$.
%
%For a surviving left-half permutation, Lemma~\ref{lem:xy-selector} combined with
%\eqref{eq:finite-first-bits} shows that output symbols $0$ and $1$
%either occupy positions $\{1,4\}$ or positions $\{2,3\}$. There are four such
%permutations of each kind.  In the first case, symbols $2$ and $3$
%occupy positions $\{2,3\}$, giving $A=2$ by
%\eqref{eq:finite-distortion}, whereas in the second case $A=0$.  Hence, among
%the eight surviving left-half permutations, four have $A=0$ and four
%have $A=2$.
%
%Similarly, among the eight surviving right-half permutations, four
%have symbols $2$ and $3$ at positions $\{6,7\}$, giving $B=2$, and
%four have them at positions $\{5,8\}$, giving $B=0$.  Hence, among the
%$64$ surviving outputs, $A-B$ equals zero for $32$ outputs, $-2$ for
%$16$, and $+2$ for $16$.  Since ties are counted as errors,\footnote{For the curious reader, we note that Lemma \ref{lem:wrong-step-75} also holds when ties are broken uniformly at random; the comparison $1/12<7/36$ to follow in \eqref{eq:pe_comparison} is then replaced by $1/18<1/9$.} this gives
%\begin{equation}
%	\widehat P_e^{(n)}(\bx,V_n)=\frac{32+16}{576}=\frac{1}{12}.
%	\label{eq:finite-Pehat}
%\end{equation}

{\bf Analysis of the $\by$-dependent selector.}
By \eqref{eq:finite-Qtilde}, this selector has $z_j=1$ at every
coordinate.  The transmitted word again satisfies the restriction for
every output: Within each half every output symbol occurs exactly once,
and $\bx$ is constant on each half, so the corresponding pairs of
counts in \eqref{eq:finite-restrict-0}--\eqref{eq:finite-restrict-1}
are equal.

We next determine when $\bar{\bx}$ satisfies the restriction.  Defining the quantities
\begin{equation}
	\nu_h=\bar x_{\pos_h(0)}-\bar x_{\pos_h(1)},
	\qquad
	\omega_h=\bar x_{\pos_h(2)}-\bar x_{\pos_h(3)},
	\qquad h\in\{L,R\},
	\label{eq:finite-DE}
\end{equation}
we have the following.

\begin{lemma}
	Under the preceding setup, for the deterministic $\by$-dependent selector in \eqref{eq:finite-Qtilde}, the competitor $\bar{\bx}$ survives the restriction in Definition \ref{def:restriction} if and only if
	\begin{equation}
		(\nu_R,\omega_R)=-(\nu_L,\omega_L).
		\label{eq:finite-DE-condition}
	\end{equation}
\end{lemma}

\begin{proof}
	We first consider the $a=0$ part of the restriction.  Since the selector
	has $z_j=1$ at every coordinate, the definition of $k_{ab}(y)$ from \eqref{eq:kaby} becomes $k_{ab}(y)=\1\{x(y)=a\}\bigl(\1\{\bar x_{\pos_L(y)}=b\}+\1\{\bar x_{\pos_R(y)}=b\}\bigr)$, and combined with $x(y)$ from \eqref{eq:finite-xy}, this gives
	\begin{align}
		\nu_L+\nu_R
		&= \bigl(\bar x_{\pos_L(0)}-\bar x_{\pos_L(1)}\bigr)
		+ \bigl(\bar x_{\pos_R(0)}-\bar x_{\pos_R(1)}\bigr) \\
		&= k_{01}(0)-k_{01}(1).
	\end{align}
	In particular, $\nu_L+\nu_R=0$ if and only if
	$k_{01}(0)=k_{01}(1)$, i.e., the counts of output symbols $0$ and $1$
	at coordinates with competitor bit $b=1$ are equal.  Since each of the
	symbols $0$ and $1$ occurs exactly twice overall, we have
	\begin{equation}
		k_{00}(y)=2-k_{01}(y),\qquad y\in\{0,1\},
		\label{eq:finite-k-complement}
	\end{equation}
	so that the condition $k_{01}(0)=k_{01}(1)$ also implies
	$k_{00}(0)=k_{00}(1)$.  Therefore, when $\nu_L+\nu_R=0$, both the
	$b=0$ and $b=1$ parts of \eqref{eq:finite-restrict-0} (for $a=0$) hold with zero
	left-hand side.
	
	Conversely, suppose that $\nu_L+\nu_R\ne0$.  If
	$|\nu_L+\nu_R|=2$, then $(k_{01}(0),k_{01}(1))$ is either $(2,0)$ or
	$(0,2)$.  Hence, $k_{01}=k_{01}(0)+k_{01}(1)=2$, and the corresponding
	discrepancy in \eqref{eq:finite-restrict-0} is $1$, whereas the
	right-hand side is $2\alpha_n<1$.  If $|\nu_L+\nu_R|=1$, then
	$(k_{01}(0),k_{01}(1))$ is one of $(1,0)$, $(0,1)$, $(2,1)$, or
	$(1,2)$.  In the first two cases we have $k_{01}=1$, and the corresponding
	discrepancy is $1/2>\alpha_n$ (for $b = 1$).  In the latter two cases, we
	find from \eqref{eq:finite-k-complement}  that the complementary pair $(k_{00}(0),k_{00}(1))$ is $(0,1)$ or $(1,0)$, so that $k_{00}=1$ and again
	the discrepancy is $1/2>\alpha_n$ (for $b=0$).
	
	Combining the above cases, the $a=0$ part of the
	restriction holds if and only if $\nu_L+\nu_R=0$. 
	By the same argument, with output symbols $2$ and $3$ in place of $0$ and
	$1$, the $a=1$ part of the restriction holds if and only if $\omega_L+\omega_R=0$. 
	Combining these two conditions gives \eqref{eq:finite-DE-condition}.
\end{proof}

We now count the outputs satisfying \eqref{eq:finite-DE-condition}.
In each half, two positions have competitor bit zero and two have
competitor bit one.  We first consider the left half, and split into two possible cases:
\begin{itemize}	
	\item {\bf Case I: $(\nu_L,\omega_L)=(0,0)$.}  In this case, output symbols $0$ and $1$ are placed at
	the two positions having the same competitor bit, and likewise for
	symbols $2$ and $3$.  There are $2^3 = 8$ such permutations: The first factor chooses whether symbols $0$ and $1$
	are placed at the two positions with competitor bit zero or at the two
	positions with competitor bit one, and the other two factors order the two
	pairs of symbols within their respective positions.
	\item {\bf Case II: $(\nu_L,\omega_L)\in\{-1,1\}^2$.}  In this case, for any prescribed value
	of this pair, the signs determine which output symbol in $\{0,1\}$, and which one in $\{2,3\}$, is placed at a position with competitor bit one.  There are two choices
	for each of these, amounting to four left-half permutations for each prescribed pair in $\{-1,1\}^2$.
	\item {\bf Absence of other cases.} We claim that there are no other possibilities for $(\nu_L,\omega_L)$ beyond Cases I and II.  Indeed, if $\nu_L=0$, then symbols $0$ and $1$ occupy the two positions having the same competitor bit, forcing symbols $2$ and $3$ to occupy the two positions having the other competitor bit, giving $\omega_L=0$.  Similarly, if $\nu_L\ne0$, then symbols $0$ and $1$ occupy positions with different
	competitor bits, forcing the same to hold for symbols $2$ and $3$, and hence
	$\omega_L\ne0$.  
\end{itemize}
By analogous reasoning, the same counts also apply in the right half.

By \eqref{eq:finite-DE-condition}, if
$(\nu_L,\omega_L)=(0,0)$ then the right half must also have
$(\nu_R,\omega_R)=(0,0)$, giving $8^2$ relevant permutations (i.e., pairs of half-permutations).  
Otherwise, there are four possible values of
$(\nu_L,\omega_L)\in\{-1,1\}^2$, and for each one the value of
$(\nu_R,\omega_R)$ is fixed to be its negative, with four permutations
in each half.  Consequently, the total number of permutations satisfying \eqref{eq:finite-DE-condition} is
\begin{equation}
	8^2+4\cdot 4^2=128.
	\label{eq:finite-tilde-survive}
\end{equation}
Under the all-zero distortion metric with ties counted as errors, these 128 permutations all lead to an error, giving
\begin{equation}
	\wc{P}_e^{(n)}(\bx,V_n)=\frac{128}{576}=\frac{2}{9}.
	\label{eq:finite-Petilde}
\end{equation}

%It remains to consider the distortion difference among these $128$
%outputs.  For the $8^2 = 64$ outputs having
%$(\nu_L,\omega_L)=(\nu_R,\omega_R)=(0,0)$, the same calculation as that of the $(x,y)$-dependent selector shows
%that $A-B$ equals zero for $32$ outputs, $-2$ for $16$, and $2$ for
%$16$.
%
%For each of the remaining $64$ outputs satisfying
%\eqref{eq:finite-DE-condition}, both $\omega_L$ and $\omega_R$ are non-zero.
%Hence, in each half, symbols $2$ and $3$ occupy positions with different
%competitor bits.  In the left half, \eqref{eq:finite-first-bits} shows that
%positions $\{2,3\}$ both have competitor bit one, while positions $\{1,4\}$
%both have competitor bit zero.  Therefore, exactly one of the output symbols
%$2$ and $3$ lies in positions $\{2,3\}$, giving $A=1$.  Similarly,
%\eqref{eq:finite-second-bits} shows that positions $\{6,7\}$ both have
%competitor bit zero, while positions $\{5,8\}$ both have competitor bit one.
%Hence, exactly one of the output symbols $2$ and $3$ lies in positions
%$\{6,7\}$, giving $B=1$.  As a result, all of these $64$ outputs have $A-B=0$.
%
%Therefore, among the $128$ outputs in \eqref{eq:finite-tilde-survive},
%the distortion difference is zero for $96$, $-2$ for $16$, and $2$
%for $16$.  Combining these and recalling that ties are counted as errors, we obtain
%\begin{equation}
%	\wc{P}_e^{(n)}(\bx,V_n)=\frac{96+16}{576}=\frac{7}{36}.
%	\label{eq:finite-Petilde}
%\end{equation}

{\bf Wrapping up.} Combining \eqref{eq:finite-Pehat} and \eqref{eq:finite-Petilde} gives
\begin{equation}
	\widehat P_e^{(n)}(\bx,V_n)=\frac{1}{9}<\frac{2}{9}
	=\wc{P}_e^{(n)}(\bx,V_n), \label{eq:pe_comparison}
\end{equation}
which is the opposite direction to \eqref{eq:Bal-75}, thus completing the proof of Lemma \ref{lem:wrong-step-75}.

\section{Omitted Proofs from Section \ref{sec:prop65-superposition}}

\subsection{Proof of Lemma \ref{lem:gb} (Relation Between $g_b$ and Restriction Condition)} \label{app:pf-gb}

Summing the definition of $g_b$ in \eqref{eq:gb-definition} over the $m$ pair positions gives
\begin{equation}
	\sum_{i=1}^m g_b(A_i,\overline X_i^\dagger)
	= \sum_{j=1}^{2m}
	\mathbf{1}\{X_j=0,Z_j=1,\overline X_j=b\}
	\left( \mathbf{1}\{Y_j=1\}
	-\frac{999}{3250}\mathbf{1}\{Y_j\in\{1,2\}\} \right).
	\label{eq:gb-expanded}
\end{equation}
By the definition of $S$ in \eqref{eq:S-zero}--\eqref{eq:S-one} with the sets in \eqref{eq:Y-plus-example}, we have $S(y,1 | 1)=0$ for $y\in\{1,2\}$. Hence, for every realization in the support of $S_{2m}$, we have the following implication:
\begin{equation}
	Z_j=1,\quad Y_j\in\{1,2\}
	\quad\Longrightarrow\quad
	X_j=0.
	\label{eq:S-support-property}
\end{equation}
This implies that on the support of $S_{2m}$, the condition $X_j=0$ in \eqref{eq:gb-expanded} is redundant and can be removed, giving
\begin{equation}
	\sum_{i=1}^m g_b(A_i,\overline X_i^\dagger)
	= \sum_{j:Z_j=1}
	\1\{\overline X_j=b,Y_j=1\} -
	\frac{999}{3250} \sum_{j:Z_j=1} \1\{\overline X_j=b,Y_j\in\{1,2\}\}. \label{eq:two-sums}
\end{equation}
Then, using the fact that $x(1)=x(2)=0$ (due to \eqref{eq:x(y)} and \eqref{eq:Y-plus-example}), the first term in \eqref{eq:two-sums} is 
\begin{equation}
	\sum_{j:Z_j=1}\mathbf{1}\{\overline X_j=b,Y_j=1,x(1)=0\}=k_{0b}(1),
\end{equation}
by the definition of $k_{ab}(y)$ in \eqref{eq:kaby}, 
and the second summation similarly gives
\begin{equation}
	\sum_{j:Z_j=1}\mathbf{1}\{\overline X_j=b,Y_j\in\{1,2\}\}
	= k_{0b}(1)+k_{0b}(2)
	= k_{0b}.
\end{equation}
Substituting these identities into \eqref{eq:two-sums} gives \eqref{eq:gb-count-interpretation}.

It remains to establish the second claim regarding the $y=2$ part.  Using $k_{0b}=k_{0b}(1)+k_{0b}(2)$ and  $\Delta V_0(2)=1 - \Delta V_0(1)$ (since $\Delta V_0(0) = 0$),  we have
\begin{equation}
	k_{0b}(2)-\Delta V_0(2)k_{0b} = -( k_{0b}(1)-\Delta V_0(1)k_{0b} ),
\end{equation}
which shows that the absolute-value constraints in \eqref{eq:Bal-restriction} corresponding to $y=1$ and $y=2$ are identical as claimed.

\subsection{Proof of Lemma \ref{lem:types} (Method of Types Analysis)} \label{app:pf-types}

	Throughout the proof, we let $\wh{P}_{(\cdot)}$ denote the joint empirical distribution of the variables indicated in the subscript. 

Fix a compatible joint type $\wt P_{A\overline X^\dagger}$, where
$A=(U,X^\dagger,Y^\dagger,Z^\dagger)$. By the superposition coding construction described in Section \ref{sec:claimed},
the transmitted pair sequence $(\bU,\bX^\dagger)$ has exact type $Q_{UX^2}$,
and the competing sequence $\overline{\bX}^\dagger$ has exact conditional
type $Q_{X^2|U}$ given $\bU$. 

We first consider the marginal of $(U,X^\dagger,Y^\dagger,Z^\dagger)$ in the absence of $\overline{\bX}^\dagger$. If the pair channel were memoryless according to $m$ uses of $S^2$, standard conditional type-class bounds would give the following (e.g., similarly to \cite[Eq.~(A.12)]{ScarlettThesis}):
\begin{equation}
	\Pr\left[\wh P_{UX^\dagger Y^\dagger Z^\dagger}=\wt P_{UX^\dagger Y^\dagger Z^\dagger}\,\Big|\,(\boldsymbol U,\boldsymbol X^\dagger)\right]
	= \exp\left(-mD\big(\wt P_{Y^\dagger Z^\dagger|UX^\dagger}\big\|S^2\big|Q_{UX^2}\big)+O(\log m)\right).
	\label{eq:output-type-probability}
\end{equation}
The actual channel $S_{2m}$ is obtained by conditioning on the exact
single-coordinate conditional type of $(Y,Z)$ given $X$. For types
satisfying \eqref{eq:F-exact-S}, this conditioning changes
\eqref{eq:output-type-probability} by only a polynomial factor
(analogous to \eqref{eq:type-prob-lower}), and hence only impacts the $O(\log m)$ term in \eqref{eq:output-type-probability}.

Next, condition on a realization
$(\boldsymbol U,\boldsymbol X^\dagger,\boldsymbol Y^\dagger,
\boldsymbol Z^\dagger)$ of the specified type. The competing sequence
$\overline{\boldsymbol X}^\dagger$ is drawn uniformly from the
conditional type class corresponding to $Q_{X^2|U}$. A standard argument via the ratio of conditional type class sizes therefore gives the following  (e.g., similarly to \cite[Eq.~(A.17)]{ScarlettThesis}):
\begin{equation}
	\Pr\left[\wh P_{A\overline X^\dagger}=\wt P_{A\overline X^\dagger}\,\middle|\,\boldsymbol U,\boldsymbol X^\dagger,\boldsymbol Y^\dagger,\boldsymbol Z^\dagger\right]
	= \exp\left(-m I_{\wt P}(\overline X^\dagger;X^\dagger,Y^\dagger,Z^\dagger|U)+O(\log m)\right).
	\label{eq:competitor-type-probability}
\end{equation}
Using $\wt P_{U\overline X^\dagger}=Q_{UX^2}$, we have
%$Q_{X^2|U}=\wt P_{\overline X^\dagger|U}$, and hence
$I_{\wt P}\big(\overline X^\dagger;X^\dagger,Y^\dagger,Z^\dagger\mid U\big) = D\big(\wt P_{\overline X^\dagger|A}\big\|Q_{X^2|U}\big|\wt P_A\big)$, which implies
\begin{align}
	&D\big(\wt P_{Y^\dagger Z^\dagger|UX^\dagger}\big\|S^2\big|Q_{UX^2}\big) +
	I_{\wt P}\big(\overline X^\dagger;X^\dagger,Y^\dagger,Z^\dagger\mid U\big) \notag\\
	&\qquad = D\big(\wt P_{Y^\dagger Z^\dagger|UX^\dagger}\big\|S^2\big|Q_{UX^2}\big)
	+ D\big(\wt P_{\overline X^\dagger|A}\big\|Q_{X^2|U}\big|\wt P_A\big)\\
	&\qquad = D\big(\wt P_{A\overline X^\dagger}\big\|B\big),
\end{align}
where the last step uses the chain rule for KL divergence and the fact that $\wt{P}_{A\overline X^\dagger}$ and $B$ have matching $(U,X^{\dagger})$ marginals (namely $Q_{UX^2}$). 
Hence, the sum of the two exponents in  \eqref{eq:output-type-probability}--\eqref{eq:competitor-type-probability} is $D(\wt P_{A\overline X^\dagger}\|B)$, giving an overall probability of
\begin{equation}
	\exp\left(
	-mD(\wt P_{A\overline X^\dagger}\|B)
	+O(\log m)
	\right)
	\label{eq:joint-type-probability}
\end{equation}
associated with a single joint type $\wt{P}_{A\overline X^\dagger}$.  Since there are only polynomially many pair types, summing \eqref{eq:joint-type-probability} over the compatible types satisfying the necessary conditions \eqref{eq:F-metric}--\eqref{eq:F-restriction} (and also satisfying \eqref{eq:F-transmitted-marginal}--\eqref{eq:F-exact-S} by construction) proves \eqref{eq:one-satellite-bound}.

\subsection{Proof of Lemma \ref{lem:dual-cert} (Dual Certificate)} \label{app:pf-dual-cert}

We use the following well-known variational inequality, which holds for any real-valued function $F$ (e.g., see \cite[Theorem 4.6]{PolyanskiyWu2025}):
\begin{equation}
	D(P\|B)
	\geq \E_P[F]-\log\E_B[e^F].
	\label{eq:KL-variational}
\end{equation}
We apply this with
\begin{align}
	F(A,\overline X^\dagger)
	& =a(U,X^\dagger)+\bar a(U,\overline X^\dagger)
	+\sum_{x,y,z}h_{xyz}\phi_{xyz}(A) \notag\\
	&\quad
	+\mu_0g_0(A,\overline X^\dagger)
	+\mu_1g_1(A,\overline X^\dagger)
	+s\,\psi(A,\overline X^\dagger),
	\label{eq:F-dual}
\end{align}
where $s\geq0$.  We claim that the constraints defining $\cF_{\alpha^*}$ in \eqref{eq:F-transmitted-marginal}--\eqref{eq:F-restriction} give 
\begin{align}
	\cE_{\alpha^*}
	& \geq \E_{Q_{UX^2}}[a(U,X^\dagger)]
	+\E_{Q_{UX^2}}[\bar a(U,\overline X^\dagger)]
	+\sum_{x,y,z}h_{xyz}r_{xyz} \notag\\
	&\quad
	-2\alpha^*(|\mu_0|+|\mu_1|)
	-\log\E_B[e^{F(A,\overline X^\dagger)}].
	\label{eq:Ealpha-dual}
\end{align}
To see this, first note that for $\wt P\in\cF_{\alpha^*}$, it holds that
$\E_{\wt P}[F] =
\E_{\wt P}[a(U,X^\dagger)]
+\E_{\wt P}[\bar a(U,\overline X^\dagger)]
+\sum_{x,y,z}h_{xyz}\E_{\wt P}[\phi_{xyz}(A)]
+\mu_0\E_{\wt P}[g_0]
+\mu_1\E_{\wt P}[g_1]
+s\E_{\wt P}[\psi]$.
The first three terms reduce to the first three terms in \eqref{eq:Ealpha-dual} by the equality constraints defining $\cF_{\alpha^*}$. Moreover, $s\E_{\wt P}[\psi]\geq 0$ due to \eqref{eq:F-metric} and $s\geq0$, while
$\mu_b\E_{\wt P}[g_b]\geq -2\alpha^*|\mu_b|$ for $b\in\{0,1\}$ by the restriction constraints. Combining these facts with \eqref{eq:KL-variational} and minimizing over $\wt P\in\cF_{\alpha^*}$ yields \eqref{eq:Ealpha-dual}.

We consider the following explicit choices:
\begin{itemize}
	\item The auxiliary cost functions $a$ and $\bar a$ are chosen as follows (entries marked by a dash have zero probability under $Q_{UX^2}$ and are irrelevant):
	\begin{equation}
		\begin{array}{c|rrrr|rrrr}
			&\multicolumn{4}{c|}{a(U,X^\dagger)}
			&\multicolumn{4}{c}{\bar a(U,\overline X^\dagger)}\\
			&00&01&10&11&00&01&10&11\\ \hline
			U=0&0&-1.848&-1.848&-&0&-0.614&-0.614&-\\
			U=1&-&-&-&-3.189&-&-&-&-3.523
		\end{array}
		\label{eq:Ealpha-costs}
	\end{equation}
	\item The $h$ function is chosen as follows:
	\begin{equation}
		\begin{array}{c|rrrrrr}
			&h_{x,0,0}&h_{x,0,1}&h_{x,1,0}&h_{x,1,1}&h_{x,2,0}&h_{x,2,1}\\ \hline
			x=0&0&0&0.125&8.207&0&-3.642\\
			x=1&2.057&2.057&-18.655&0&3.669&0
		\end{array}.
		\label{eq:Ealpha-h}
	\end{equation}
	\item The remaining scalar parameters are chosen as follows:
	\begin{equation}
		\mu_0=-11.85,
		\qquad \mu_1=18.372,
		\qquad s=30.2.
		\label{eq:Ealpha-mus}
	\end{equation}
\end{itemize}
A direct evaluation of \eqref{eq:Ealpha-dual} then gives $\cE_{\alpha^*} 
\geq0.35870725 - 60.444\alpha^*$ (see \cite{ScarlettCode2026} for numerical verification), thus establishing \eqref{eq:Ealpha-certificate}.

\section*{AI Declaration}

The main goal of this paper was to identify at least one specific step in \cite{Balakirsky1995} that appears to be (i) incorrect, and (ii) significant enough that it is unlikely to admit a local repair.  This goal was pursued with substantial assistance from an AI agent (OpenAI ChatGPT, models 5.6 Sol and Astra).  It initially identified the issue discussed in Section~3.4.1, but after further analysis I was only convinced of (i) and not (ii).  Subsequent discussion identified the steps discussed in Section~\ref{sec:false-steps} as considerably stronger candidates for ``major flaws'' satisfying both (i) and (ii).

AI was then used to generate initial versions of several mathematical arguments and the initial version of the accompanying code (the latter using OpenAI Codex).  I checked all of these closely, substantially edited, expanded, and rewrote the proofs, formed the overall paper, and further improved the code through both AI interaction and manual editing.  As the sole author, I take responsibility for all claims in this paper and their correctness, and for the accompanying code.

\bibliographystyle{IEEEtran}
\bibliography{refs}

\end{document}